\documentstyle[12pt]{article}

\font\twlgot =eufm10 scaled \magstep1 \font\egtgot =eufm8
\font\sevgot =eufm7 \font\twlmsb =msbm10 scaled \magstep1
\font\egtmsb =msbm8 \font\sevmsb =msbm7

\newfam\gotfam
\def\pgot{\fam\gotfam\twlgot}
\textfont\gotfam\twlgot \scriptfont\gotfam\egtgot
\scriptscriptfont\gotfam\sevgot
\def\got{\protect\pgot}
\newfam\msbfam
\textfont\msbfam\twlmsb \scriptfont\msbfam\egtmsb
\scriptscriptfont\msbfam\sevmsb
\def\Bbb{\protect\pBbb}
\def\pBbb{\relax\ifmmode\expandafter\Bb\else\typeout{You cann't use
Bbb in text mode}\fi}
\def\Bb #1{{\fam\msbfam\relax#1}}

\newcommand{\gO}{{\got O}}
\newcommand{\gQ}{{\got T}}
\newcommand{\gU}{{\got U}}
\newcommand{\gE}{{\got E}}
\newcommand{\gA}{{\got A}}

\newcommand{\gd}{{\got d}}
\newcommand{\gP}{{\got P}}
\newcommand{\gX}{{\got X}}
\newcommand{\gS}{{\got S}}

\def\thebibliography#1{\section*{References}\list
  {[\arabic{enumi}]}{\settowidth\labelwidth{#1}\leftmargin\labelwidth
    \advance\leftmargin\labelsep
    \usecounter{enumi}}
    \def\newblock{\hskip .11em plus .33em minus .07em}
    \sloppy\clubpenalty4000\widowpenalty4000
    \sfcode`\.=1000\relax}

\def\op#1{\mathop{\fam0 #1}\limits}

\newcommand{\Ker}{{\rm Ker\,}}

\newcommand{\nm}[1]{|{#1}|}
\newcommand{\beq}{\begin{equation}}
\newcommand{\eeq}{\end{equation}}
\newcommand{\ben}{\begin{eqnarray}}
\newcommand{\een}{\end{eqnarray}}
\newcommand{\be}{\begin{eqnarray*}}
\newcommand{\ee}{\end{eqnarray*}}
\newcommand{\bea}{\begin{eqalph}}
\newcommand{\eea}{\end{eqalph}}
\newcommand{\cA}{{\cal A}}

\newcommand{\cP}{{\cal P}}

\newcommand{\cL}{{\cal L}}

\newcommand{\cE}{{\cal E}}

\newcommand{\cQ}{{\cal Q}}

\newcommand{\cS}{{\cal S}}

\newcommand{\cO}{{\cal O}}

\newcommand{\cK}{{\cal K}}
\newcommand{\bL}{{\bf L}}

\newcommand{\bE}{{\bf E}}

\newcommand{\al}{\alpha}
\newcommand{\vr}{\varrho}

\newcommand{\dl}{\delta}
\newcommand{\la}{\lambda}
\newcommand{\La}{\Lambda}
\newcommand{\f}{\phi}
\newcommand{\om}{\omega}

\newcommand{\m}{\mu}

\newcommand{\g}{\gamma}
\newcommand{\G}{\Gamma}
\newcommand{\th}{\theta}

\newcommand{\vt}{\vartheta}
\newcommand{\vf}{\varphi}
\newcommand{\up}{\upsilon}

\newcommand{\si}{\sigma}
\newcommand{\Si}{\Sigma}
\newcommand{\w}{\wedge}
\newcommand{\wt}{\widetilde}
\newcommand{\wh}{\widehat}
\newcommand{\ol}{\overline}
\newcommand{\dr}{\partial}
\newcommand{\ar}{\op\longrightarrow}

\newcommand{\ot}{\otimes}
\newcommand{\ap}{\approx}

\newcounter{eqalph}
\newcounter{equationa}
\newcounter{remark}
\newcounter{example}
\newcounter{theorem}
\newcounter{proposition}
\newcounter{lemma}
\newcounter{corollary}
\newcounter{definition}
\newenvironment{eqalph}{\stepcounter{equation}
\setcounter{equationa}{\value{equation}} \setcounter{equation}{0}

\begin{eqnarray}}{\end{eqnarray}\setcounter{equation}{\value{equationa}}}

\def\theremark{\arabic{remark}}

\def\thetheorem{\arabic{theorem}}

\newenvironment{proof}{
{\bf Proof.}}{\hfill $\Box$ \medskip}
\newenvironment{rem}{\refstepcounter{remark}\medskip\noindent{\bf
Remark \theremark.}}{\medskip}

\newenvironment{theo}{\refstepcounter{theorem}\medskip\noindent
{\bf Theorem \thetheorem.}\it }{\medskip}
\newenvironment{prop}{\refstepcounter{theorem}\medskip\noindent
{\bf Proposition \thetheorem.}\it }{\medskip}
\newenvironment{lem}{\refstepcounter{theorem}\medskip\noindent
{\bf Lemma \thetheorem.}\it }{\medskip}
\newenvironment{cor}{\refstepcounter{theorem}\medskip\noindent
{\bf Corollary \thetheorem.}\it }{\medskip}

\newcommand{\mar}[1]{}

\begin{document}
\hbox{}

\begin{center}

{\large\bf GRADED INFINITE ORDER JET MANIFOLDS}

\bigskip

{\sc G. SARDANASHVILY}
\medskip

\begin{small}

 {\it Department of Theoretical Physics,
Moscow State University, 117234 Moscow, Russia}


\end{small}
\end{center}

\bigskip

\begin{small}
\noindent  The relevant material on differential calculus on
graded infinite order jet manifolds and its cohomology is
summarized. This mathematics provides the adequate formulation of
Lagrangian theories of even and odd variables on smooth manifolds
in terms of the Grassmann-graded variational bicomplex.


\end{small}

\section{Introduction}

Let $Y\to X$ be a smooth fiber bundle and $J^\infty Y$ the
Fr\'echet manifold of infinite order jets of its sections.  The
differential calculus on $J^\infty Y$ and its cohomology provide
the adequate mathematical description of Lagrangian theories on
$Y\to X$ in terms of the variational bicomplex
\cite{ander,jmp,tak2}. This description has been extended to
Lagrangian theories on graded manifolds in terms of the
Grassmann-graded variational bicomplex of differential forms  on a
graded infinite order jet manifold \cite{barn,jmp05,jmp05a,cmp04}.

Different geometric models of odd variables are phrased in terms
both of graded manifolds and supermanifolds. Note that graded
manifolds are characterized by sheaves on smooth manifolds, while
supermanifolds are defined by gluing of sheaves on supervector
spaces \cite{bart,book05}. Treating odd variables on smooth
manifolds, we follow the Serre--Swan theorem for graded manifolds
(Theorem \ref{v0}). It states that a graded commutative
$C^\infty(X)$-ring is isomorphic to an algebra of graded functions
on a graded manifold with a body $X$ iff it is the exterior
algebra of some projective $C^\infty(X)$-module of finite rank. By
virtues of the Batchelor theorem \cite{bart}, any graded manifold
$(Z,\gA)$ with a body $Z$ and a structure sheaf $\gA$ of graded
functions is isomorphic to a graded manifold $(Z,\gA_Q)$ modelled
over some vector bundle $Q\to Z$, i.e., its structure sheaf
$\gA_Q$ is the sheaf of sections of the exterior bundle $\w Q^*$,
where $Q^*$ is the dual of $Q\to Z$. Our goal is the following
differential bigraded algebra (henceforth DBGA)
$\cS^*_\infty[F;Y]$ and its relevant cohomology.

Let $F\to Y\to X$ be a composite bundle where $F\to Y$ is a vector
bundle. Jet manifolds $J^rF$ of $F\to X$ are also vector bundles
over $J^rY$. Let $(J^rY,\gA_r)$ be a graded manifold modelled over
$J^rF\to J^rY$, and let $\cS^*_r[F;Y]$ be the DBGA of
Grassmann-graded differential forms on the graded manifold
$(J^rY,\gA_r)$. There is the inverse system of jet manifolds
\mar{j1}\beq
Y\op\longleftarrow^\pi J^1Y \longleftarrow \cdots J^{r-1}Y
\op\longleftarrow^{\pi^r_{r-1}} J^rY\longleftarrow\cdots.
\label{j1}
\eeq
Its projective limit $J^\infty Y$ is a paracompact Fr\'echet
manifold, called the infinite order jet manifold. This inverse
system yields the direct system of DBGAs
\mar{j2}\beq
\cS^*[F;Y]\ar^{\pi^*} \cS^*_1[F;Y]\ar\cdots \cS^*_{r-1}[F;Y]
\op\ar^{\pi^{r*}_{r-1}}\cS^*_r[F;Y]\ar\cdots, \label{j2}
\eeq
where $\pi^{r*}_{r-1}$ the pull-back monomorphisms. Its direct
limit is the above mentioned DBGA $\cS^*_\infty[F;Y]$ of all
Grassmann-graded differential forms on graded manifolds
$(J^rY,\gA_r)$ modulo the pull-back identification. On can think
of elements of $\cS^*_\infty[F;Y]$ as being Grassmann-graded
differential forms on a graded manifold $(J^\infty Y,
\gA_\infty)$, called the graded infinite order jet manifold, whose
body is $J^\infty Y$ and the structure sheaf $\gA_\infty$ is the
sheaf of germs of elements of $\cS^*_\infty[F;Y]$.

The DBGA $\cS^*_\infty[F;Y]$ is split into the above mentioned
Grassmann-graded variational bicomplex, describing Lagrangian
theories of even and odd variables on a smooth manifold $X$.
Grassmann-graded Lagrangians and their Euler--Lagrange operators
are elements of this bicomplex. Its cohomology results in the
global first variational formula, the first Noether theorem and
defines a class of variationally trivial Lagrangians.

It should be emphasized that this description of Grassmann-graded
Lagrangian systems differs from that phrased in terms of fibered
graded manifolds \cite{hern,mont05}, but reproduces the heuristic
formulation of Lagrangian BRST theory \cite{barn,bran01}. Namely,
$(J^\infty Y, \gA_\infty)$ is a graded manifold of jets of smooth
fiber bundles, but not jets of fibered graded manifolds.

\section{Technical preliminary}

Throughout the paper, smooth manifolds are real and
finite-dimensional. They are Hausdorff and second-countable
topological spaces (i.e., have a countable base for topology).
Consequently, they are paracompact, separable (i.e., have a
countable dense subset), and locally compact topological spaces,
which are countable at infinity. Unless otherwise stated, smooth
manifolds are assumed to be connected and, consequently, arcwise
connected. It is essential for our consideration that a
paracompact smooth manifold admits the partition of unity by
smooth functions. Real-analytic manifolds are also considered as
smooth ones because they need not possess the partition of unity
by real-analytic functions.

Only proper covers $\gU=\{U_\iota\}$ of smooth manifolds are
considered, i.e., $U_\iota\neq U_{\iota'}$ if $\iota\neq \iota'$.
A cover $\gU'$ is said to be a refinement of a cover $\gU$ if, for
each $U'\in\gU'$, there exists $U\in\gU$ such that $U'\subset U$.
For any cover $\gU$ of an $n$-dimensional smooth manifold $X$,
there exists a countable atlas $\{(U'_\iota,\vf_\iota)\}$ of $X$
such that: (i) the cover $\{U'_\iota\}$ refines $\gU$, (ii)
$\vf_\iota(U'_\iota)=\Bbb R^n$, and (iii) the closure $\ol
U'_\iota$ of any $U_\iota$ is compact \cite{greub}.

Let $\pi:Y\to X$ be a smooth fiber bundle. There exist the
following particular covers of $X$ which one can choose for its
bundle atlas \cite{greub}.

(i) There is a bundle atlas of $Y$ over a countable cover $\gU$ of
$X$ where each member $U_\iota$ of $\gU$ is a domain (i.e., a
contractible open subset) and its closure $\ol U_\iota$ is
compact.

(ii) There exists a bundle atlas of $Y$ over a finite cover of
$X$. Indeed, let $\Psi$ be a bundle atlas of $Y\to X$ over a cover
$\gU$ of $X$. For any cover $\gU$ of a manifold $X$, there exists
its refinement $\{U_{ij}\}$, where $j\in\Bbb N$ and $i$ runs
through a finite set such that $U_{ij}\cap U_{ik}=\emptyset$,
$j\neq k$. Let $\{(U_{ij}, \psi_{ij})\}$ be the corresponding
bundle atlas of a fiber bundle $Y\to X$. Then $Y$ has the finite
bundle atlas $U_i\op=\op\cup_j U_{ij}$, $
\psi_i(x)\op=\psi_{ij}(x)$, $x\in U_{ij}\subset U_i$, whose
members$U_i$ however need not be contractible and connected.

Without a loss of generality, we further assume that a cover $\gU$
for a bundle atlas of $Y\to X$ is also a cover for a manifold
atlas of its base $X$. Given such an atlas, a fiber bundle $Y$ is
provided with the associated bundle coordinates $(x^\la,y^i)$
where $(x^\la)$ are coordinates on $X$.

Given a manifold $X$, its tangent and cotangent bundles $TX$ and
$T^*X$ are endowed with the bundle coordinates $(x^\la,\dot
x^\la)$ and $(x^\la,\dot x_\la)$ with respect to holonomic frames
$\{\dr_\la\}$ and $\{dx^\la\}$, respectively. Given a smooth
bundle $Y\to X$, its vertical tangent and cotangent bundles $VY$
and $V^*Y$ are provided with the bundle coordinates
$(x^\la,y^i,\dot y^i)$ and $(x^\la, y^i,\ol y_i)$, respectively.

By $\La=(\la_1...\la_k)$, $|\La|=k$, $\la+\La=(\la\la_1...\la_k)$.
are denoted symmetric multi-indices. Summation over a multi-index
$\La$ means separate summation over each its index $\la_i$. The
notation
\mar{j10}\beq
d_\la = \dr_\la + \op\sum_{0\leq |\La|}
y^i_{\la+\La}\dr_i^\La,\qquad
 d_\La=d_{\la_r}\circ\cdots\circ d_{\la_1}, \label{j10}
\eeq
stands for the total derivatives.

\section{Finite order jet manifolds}

Given a smooth fiber bundle $Y\to X$, its $r$-order jet $j^r_xs$
 is defined as the equivalence class of sections $s$ of $Y$
identified by their $r+1$ terms of their Taylor series at a point
$x\in X$. The disjoint  union $J^rY=\op\bigcup_{x\in X}j^r_xs$ of
these jets is a smooth manifold provided with the adapted
coordinates
\be
(x^\la, y^i, y^i_\La)_{|\La|\leq r},\qquad  (x^\la, y^i_\La)\circ
s= (x^\la, \dr_\La s^i(x)), \qquad {y'}^i_{\la+\La}=\frac{\dr
x^\m}{\dr x'^\la}d_\m y'^i_\La.
\ee
For the sake of brevity, the index $r=0$ further stands for $Y$.
The jet manifolds of $Y\to X$ form the inverse system (\ref{j1})
where $\pi^r_{r-1}$, $r>0$, are affine bundles.

Given fiber bundles $Y$ and $Y'$ over $X$, every bundle morphism
$\Phi:Y\to Y'$ over a diffeomorphism $f$ of $X$ admits the
$r$-order jet prolongation to the morphism of the $r$-order jet
manifolds
\be
J^r\Phi: J^rY\ni j^r_xs\mapsto j^r_{f(x)}(\Phi\circ s\circ f^{-1})
\in  J^rY'.
\ee
If $\Phi$ is an injection or surjection, so is $J^r\Phi$. It
preserves an algebraic structure. If $Y\to X$ is a vector bundle,
$J^rY\to X$ is also a vector bundle. If $Y\to X$ is an affine
bundle modelled over a vector bundle $\ol Y\to X$, then $J^rY\to
X$ is an affine bundle modelled over the vector bundle $J^r\ol
Y\to X$.

Every section $s$ of a fiber bundle $Y\to X$ admits the $r$-order
jet prolongation  to the section $(J^rs)(x)= j^r_xs$ of the jet
bundle $J^rY\to X$.

Every exterior form $\f$ on the jet manifold $J^kY$ gives rise to
the pull-back form $\pi^{k+i}_k{}^*\f$ on the jet manifold
$J^{k+i}Y$. Let $\cO_k^*$ be the differential graded algebra
(henceforth DGA) of exterior forms on the jet manifold $J^kY$. We
have the direct system of DGAs
\mar{5.7}\beq
\cO^*X\op\longrightarrow^{\pi^*} \cO^*Y
\op\longrightarrow^{\pi^1_0{}^*} \cO_1^* \longrightarrow \cdots
\cO_{r-1}^*\op\longrightarrow^{\pi^r_{r-1}{}^*}
 \cO_r^* \longrightarrow\cdots. \label{5.7}
\eeq

Every projectable vector field $u=u^\m\dr_\m +u^i\dr_i$ on a fiber
bundle $Y\to X$ has the $k$-order jet prolongation onto $J^kY$ to
the vector field
\mar{j12}\beq
 j^k u = u^\la \dr_\la + u^i\dr_i + \op\sum_{0<|\La|\leq k}
[d_\La(u^i-y^i_\m u^\m)+y^i_{\m+\La}u^\m]. \label{j12}
\eeq

Jet manifold provides the conventional language of theory of
nonlinear differential equations and differential operators on
fiber bundles \cite{bry,kras}. A $k$-order differential equation
on a fiber bundle $Y\to X$ is defined as a closed subbundle $\gE$
of the jet bundle $J^kY\to X$. Its classical solution is a (local)
section $s$ of $Y\to X$ whose $k$-order jet prolongation $J^ks$
lives in $\gE$.

Differential equations can come from differential operators. Let
$E\to X$ be a vector bundle coordinated by $(x^\la,v^A)$,
$A=1,\ldots,m$. A bundle morphism $\cE: J^kY\to E$ over $X$ is
called a $k$-order differential operator on a fiber bundle $Y\to
X$. It sends each section $s$ of $Y\to X$ onto the section
$(\cE\circ J^ks)^A(x)$ of the vector bundle $E\to X$. Let us
suppose that the canonical zero section $\wh 0(X)$ of the vector
bundle $E\to X$ belongs to $\cE(J^kY)$. Then the kernel of a
differential operator $\cE$ is defined as $\Ker\cE =\cE^{-1}(\wh
0(X))\subset J^kY$. If $\Ker\cE$ is a closed subbundle of the jet
bundle $J^kY\to X$, it is a $k$-order differential equation,
associated to the differential operator $\cE$. For instance, the
kernel of an Euler--Lagrange operator need not be a closed
subbundle. Therefore, it may happen that associated
Euler--Lagrange equations are not a differential equation in a
strict sense.

\section{Infinite order jet manifold}

Given the inverse system (\ref{j1}) of jet manifolds, its
projective limit $J^\infty Y$ is defined as a minimal set such
that there exist surjections
\mar{5.74}\beq
\pi^\infty: J^\infty Y\to X, \quad \pi^\infty_0: J^\infty Y\to Y,
\quad \quad \pi^\infty_k: J^\infty Y\to J^kY, \label{5.74}
\eeq
obeying the commutative diagrams
$\pi^\infty_r=\pi^k_r\circ\pi^\infty_k$ for any admissible $k$ and
$r<k$. A projective limit of the inverse system (\ref{j1}) always
exists. It consists of those elements $(\ldots,z_r,\ldots,z_k,
\ldots)$, $z_r\in J^rY$, $z_k\in J^kY$, of the Cartesian product
$\op\prod_k J^kY$ which obey the relations $z_r=\pi^k_r(z_k)$ for
all $k>r$. One can think of elements of $J^\infty Y$ as being
infinite order jets of sections of $Y\to X$ identified by their
Taylor series at points of $X$.

The set $J^\infty Y$ is provided with the projective limit
topology. This is the coarsest topology such that the surjections
$\pi^\infty_r$ (\ref{5.74}) are continuous. Its base consists of
inverse images of open subsets of $J^rY$, $r=0,\ldots$, under the
mappings $\pi^\infty_r$. With this topology, $J^\infty Y$ is a
paracompact Fr\'echet (complete metrizable, but not Banach)
manifold modelled on a locally convex vector space of formal
number series $\{a^\la,a^i,a^i_\la,\cdots\}$ \cite{tak2}.
Moreover, the surjections $\pi^\infty_r$ are open maps, i.e,
$J^\infty Y\to J^rY$ are topological bundles. A bundle coordinate
atlas $\{U_Y,(x^\la,y^i)\}$ of $Y\to X$ provides $J^\infty Y$ with
the manifold coordinate atlas
\mar{j3}\beq
\{(\pi^\infty_0)^{-1}(U_Y), (x^\la, y^i_\La)\}_{0\leq|\La|},
\qquad {y'}^i_{\la+\La}=\frac{\dr x^\m}{\dr x'^\la}d_\m y'^i_\La.
\label{j3}
\eeq

It is essential for our consideration that $Y$ is a strong
deformation retract of $J^\infty Y$ \cite{ander,jmp} (see Appendix
A). This result follows from the fact that a base of any affine
bundle is a strong deformation retract of its total space.
Consequently, a fiber bundle $Y$ is a strong deformation retract
of any finite order jet manifold $J^rY$. Therefore by virtue of
the Vietoris--Begle theorem \cite{bred}, there are isomorphisms
\mar{j19}\beq
H^*(J^\infty Y,\Bbb R)=H^*(J^rY,\Bbb R)=H^*(Y,\Bbb R) \label{j19}
\eeq
of cohomology groups of $J^\infty Y$, $J^rY$, $0<r$ and $Y$ with
coefficients in the constant sheaf $\Bbb R$.

Though $J^\infty Y$ fails to be a smooth manifold, one can
introduce the differential calculus on $J^\infty Y$ as follows.
Let us consider the direct system (\ref{5.7}) of DGAs. Its direct
limit $\cO^*_\infty$ exists, and consists of all exterior forms on
finite order jet manifolds modulo the pull-back identification. It
is a DGA, inheriting the DGA operations of $\cO^*_r$
\cite{massey}.

\begin{theo} \label{j4} \mar{j4} The cohomology $H^*(\cO_\infty^*)$ of
the de Rham complex
\mar{5.13} \beq
0\longrightarrow \Bbb R\longrightarrow \cO^0_\infty
\op\longrightarrow^d\cO^1_\infty \op\longrightarrow^d \cdots
\label{5.13}
\eeq
of the DGA $\cO^*_\infty$ equals the de Rham cohomology of a fiber
bundle $Y$ \cite{ander}.
\end{theo}

\begin{proof}
By virtue of the well-known theorem, the operation of taking
homology groups of cochain complexes commutes with the passage to
a direct limit \cite{massey}. Since the DGA $\cO^*_\infty$ is a
direct limit of DGAs $\cO^*_r$, its cohomology is isomorphic to
the direct limit of the direct system
\mar{j5}\beq
 H^*_{DR}(Y)\longrightarrow H^*_{DR}(J^1Y)\longrightarrow
\cdots H^*_{DR}(J^{r-1}Y)\longrightarrow
H^*_{DR}(J^rY)\longrightarrow \cdots \label{j5}
\eeq
of the de Rham cohomology groups $H^*_{DR}(J^rY)=H^*(\cO^*_r)$ of
finite order jet manifolds $J^rY$. By virtue of the de Rham
theorem \cite{hir}, the de Rham cohomology $H^*_{DR}(J^rY)$ of
$J^rY$ equals its cohomology $H^*(J^rY,\Bbb R)$ with coefficients
in the constant sheaf $\Bbb R$. Since $Y$ is a strong deformations
retract of $J^rY$, this cohomology coincides with the cohomology
$H^*(Y,\Bbb R)$ of $Y$. Consequently, the direct limit of the
direct system (\ref{j5}) is the de Rham cohomology $H^*(Y,\Bbb
R)=H_{DR}^*(Y)$ of $Y$.
\end{proof}

\begin{cor} \label{j21} \mar{j21}
Any closed form $\f\in \cO^*_\infty$ is decomposed into the sum
$\f=\si +d\xi$, where $\si$ is a closed form on $Y$.
\end{cor}

One can think of elements of $\cO_\infty^*$ as being differential
forms on the infinite order jet manifold $J^\infty Y$ as follows.
Let $\gO^*_r$ be the sheaf of germs of exterior forms on  $J^rY$
and $\ol\gO^*_r$ the canonical presheaf of local sections of
$\gO^*_r$ (we follow the terminology of \cite{hir}). Since
$\pi^r_{r-1}$ are open maps, there is the direct  system of
presheaves
\be
\ol\gO^*_0 \op\longrightarrow^{\pi^1_0{}^*} \ol\gO_1^* \cdots
\op\longrightarrow^{\pi^r_{r-1}{}^*}
 \ol\gO_r^* \longrightarrow\cdots.
\ee
Its direct limit $\ol\gO^*_\infty$ is a presheaf of DGAs on
$J^\infty Y$. Let $\gQ^*_\infty$ be the sheaf of DGAs of germs of
$\ol\gO^*_\infty$ on $J^\infty Y$. The structure module
$\cQ^*_\infty=\G(\gQ^*_\infty)$ of global sections of
$\gQ^*_\infty$ is a DGA such that, given an element $\f\in
\cQ^*_\infty$ and a point $z\in J^\infty Y$, there exist an open
neighbourhood $U$ of $z$ and an exterior form $\f^{(k)}$ on some
finite order jet manifold $J^kY$ so that $\f|_U=
\pi^{\infty*}_k\f^{(k)}|_U$. Therefore, there is the DGA
monomorphism $\cO^*_\infty  \to\cQ^*_\infty$. It should be
emphasized that the paracompact space $J^\infty Y$ admits a
partition of unity by elements of the ring $\cQ^0_\infty$, but not
$\cO^0_\infty$.

Since elements of the DGA $\cQ^*_\infty$ are locally exterior
forms on finite order jet manifolds, the following Poincar\'e
lemma holds.

\begin{lem} \label{j8} \mar{j8}
For closed element $\f\in \cQ^*_\infty$, there exists a
neighbourhood $U$ of each point $z\in J^\infty Y$ such that
$\f|_U$ is exact.
\end{lem}

\begin{theo} \label{j6} \mar{j6} The cohomology $H^*(\cQ_\infty^*)$ of
the de Rham complex
\mar{5.13'} \beq
0\longrightarrow \Bbb R\longrightarrow \cQ^0_\infty
\op\longrightarrow^d\cQ^1_\infty \op\longrightarrow^d \cdots\,.
\label{5.13'}
\eeq
of the DGA $\cQ^*_\infty$ equals the de Rham cohomology of a fiber
bundle $Y$ \cite{tak2}.
\end{theo}

\begin{proof}
Let us consider the de Rham complex of sheaves
\mar{j7} \beq
0\longrightarrow \Bbb R\longrightarrow \gQ^0_\infty
\op\longrightarrow^d\gQ^1_\infty \op\longrightarrow^d \cdots
\label{j7}
\eeq
on $J^\infty Y$. By virtue of Lemma \ref{j8}, it is exact at all
terms, except $\Bbb R$. Being the sheaves of
$\cQ^0_\infty$-modules, the sheaves $\gQ^r_\infty$ are fine and,
consequently acyclic because the paracompact space $J^\infty Y$
admits the partition of unity by elements of the ring
$\cQ^0_\infty$ \cite{hir}. Thus, the complex (\ref{j7}) is a
resolution of the constant sheaf $\Bbb R$ on $J^\infty Y$. In
accordance with the abstract de Rham theorem (see Appendix B),
cohomology $H^*(\cQ_\infty^*)$ of the complex (\ref{5.13'}) equals
the cohomology $H^*(J^\infty Y, \Bbb R)$ of $J^\infty Y$ with
coefficients in the constant sheaf $\Bbb R$. Since $Y$ is a strong
deformation retract of $J^\infty Y$, we obtain
\be
H^*(\cQ_\infty^*)=H^*(J^\infty Y, \Bbb R)=H^*(Y, \Bbb
R)=H^*_{DR}(Y).
\ee
\end{proof}

Due to the monomorphism $\cO^*_\infty  \to\cQ^*_\infty$, one can
restrict $\cO^*_\infty$ to the coordinate chart (\ref{j3}) where
horizontal forms $dx^\la$ and contact one-forms
$\th^i_\La=dy^i_\La -y^i_{\la+\La}dx^\la$ make up a local basis
for the $\cO^0_\infty$-algebra $\cO^*_\infty$. Though $J^\infty Y$
is not a smooth manifold, elements of $\cO^*_\infty$ are exterior
forms on finite order jet manifolds and, therefore, their
coordinate transformations are smooth. Moreover, there is the
canonical decomposition $\cO^*_\infty=\oplus\cO^{k,m}_\infty$ of
$\cO^*_\infty$ into $\cO^0_\infty$-modules $\cO^{k,m}_\infty$ of
$k$-contact and $m$-horizontal forms together with the
corresponding projectors
\be
h_k:\cO^*_\infty\to \cO^{k,*}_\infty, \qquad h^m:\cO^*_\infty\to
\cO^{*,m}_\infty.
\ee
Accordingly, the exterior differential on $\cO_\infty^*$ is split
into the sum $d=d_H+d_V$ of the total and vertical differentials
\be
&& d_H\circ h_k=h_k\circ d\circ h_k, \qquad d_H\circ
h_0=h_0\circ d, \qquad d_H(\f)= dx^\la\w d_\la(\f), \\
&& d_V \circ h^m=h^m\circ d\circ h^m, \qquad d_V(\f)=\th^i_\La \w
\dr^\La_i\f, \qquad \f\in\cO^*_\infty,
\ee
such that $d_H\circ d_H=0$, $d_V\circ d_V=0$, $d_H\circ
d_V+d_V\circ d_H=0$. These differentials make $\cO^{*,*}_\infty$
into a bicomplex.

Let $\vt\in\gd\cO^0_\infty$ be the $\cO^0_\infty$-module of
derivations of the $\Bbb R$-ring $\cO^0_\infty$.

\begin{prop} \label{g62} \mar{g62}
The derivation module $\gd\cO^0_\infty$ is isomorphic to the
$\cO^0_\infty$-dual $(\cO^1_\infty)^*$ of the module of one-forms
$\cO^1_\infty$.
\end{prop}

\begin{proof} At first, let us show that $\cO^*_\infty$ is generated by elements
$df$, $f\in \cO^0_\infty$. It suffices to justify that any element
of $\cO^1_\infty$ is a finite $\cO^0_\infty$-linear combination of
elements $df$, $f\in \cO^0_\infty$. Indeed, every
$\f\in\cO^1_\infty$ is an exterior form on some finite order jet
manifold $J^rY$.  By virtue of the Serre--Swan theorem extended to
non-compact manifolds \cite{book05,ren}, the
$C^\infty(J^rY)$-module $\cO^1_r$ of one-forms on $J^rY$ is a
projective module of finite rank, i.e., $\f$ is represented by a
finite $C^\infty(J^rY)$-linear combination of elements $df$, $f\in
C^\infty(J^rY)\subset \cO^0_\infty$.  Any element $\Phi\in
(\cO^1_\infty)^*$ yields a derivation $\vt_\Phi(f)=\Phi(df)$ of
the $\Bbb R$-ring $\cO^0_\infty$. Since the module $\cO^1_\infty$
is generated by elements $df$, $f\in \cO^0_\infty$, different
elements of $(\cO^1_\infty)^*$ provide different derivations of
$\cO^0_\infty$, i.e., there is a monomorphism $(\cO^1_\infty)^*\to
\gd\cO^0_\infty$. By the same formula, any derivation $\vt\in
\gd\cO^0_\infty$ sends $df\mapsto \vt(f)$ and, since
$\cO^0_\infty$ is generated by elements $df$, it defines a
morphism $\Phi_\vt:\cO^1_\infty\to \cO^0_\infty$. Moreover,
different derivations $\vt$ provide different morphisms
$\Phi_\vt$. Thus, we have a monomorphism and, consequently, an
isomorphism $\gd\cO^0_\infty\to (\cO^1_\infty)^*$.
\end{proof}

The proof of Proposition \ref{g62} gives something more. The DGA
$\cO^*_\infty$ is a minimal Chevalley--Eilenberg differential
calculus over the $\Bbb R$-ring $\cO^0_\infty$ of smooth real
functions on finite order jet manifolds of $Y\to X$.

\begin{rem}
Let $\cK$ be a commutative ring and $\cA$ a commutative
$\cK$-ring. The module $\gd\cA$ of derivations of $\cA$ is a Lie
$\cK$-algebra. The Chevalley--Eilenberg complex of the Lie algebra
$\gd\cA$ with coefficients in the ring $\cA$ contains a subcomplex
of $\cA$-multilinear skew-symmetric maps \cite{book05}. It is
called the Chevalley--Eilenberg differential calculus over a
$\cK$-ring $\cA$. The minimal Chevalley--Eilenberg calculus is
generated by monomials $a_0da_1\w\cdots\w da_k$, $a_i\in\cA$. For
instance, the DGA of exterior forms on a smooth manifold $Z$ is
the minimal Chevalley--Eilenberg differential calculus over the
$\Bbb R$-ring $C^\infty(Z)$.
\end{rem}

Restricted to a coordinate chart (\ref{j3}), $\cO^1_\infty$ is a
free $\cO^0_\infty$-module generated by the exterior forms
$dx^\la$, $\th^i_\La$. Since $\gd\cO^0_\infty=(\cO^1_\infty)^*$,
any derivation of the $\Bbb R$-ring $\cO^0_\infty$ takes the
coordinate form
\mar{g3}\beq
\vt=\vt^\la \dr_\la + \vt^i\dr_i + \op\sum_{0<|\La|}\vt^i_\La
\dr^\La_i, \label{g3}
\eeq
where $\dr^\La_i(s_\Si^j)=\dr^\La_i\rfloor
ds_\Si^j=\dl_i^j\dl^\La_\Si$ up to permutations of multi-indices
$\La$ and $\Si$. Its coefficients $\vt^\la$, $\vt^i$, $\vt^i_\La$
are local smooth functions of finite jet order possessing the
transformation law
\be
\vt'^\la=\frac{\dr x'^\la}{\dr x^\m}\vt^\m, \qquad
\vt'^i=\frac{\dr y'^i}{\dr y^j}\vt^j + \frac{\dr y'^i}{\dr
x^\m}\vt^\m, \qquad \vt'^i_\La=\op\sum_{|\Si|\leq|\La|}\frac{\dr
y'^i_\La}{\dr y^j_\Si}\vt^j_\Si + \frac{\dr y'^i_\La}{\dr
x^\m}\vt^\m.
\ee

Extended to the DGA $\cO^*_\infty$, the interior product obeys the
rule
\be
\vt\rfloor(\f\w\si)=(\vt\rfloor \f)\w\si
+(-1)^{|\f|}\f\w(\vt\rfloor\si).
\ee
Any derivation $\vt$ (\ref{g3}) of the ring $\cO^0_\infty$ yields
a derivation (a Lie derivative $\bL_\vt$) of the DGA
$\cO^*_\infty$ given by the relations
\be
\bL_\vt\f=\vt\rfloor d\f+ d(\vt\rfloor\f), \qquad
\bL_\vt(\f\w\f')=\bL_\vt(\f)\w\f' +\f\w\bL_\vt(\f').
\ee
In particular, the total derivatives (\ref{j10}) are defined as
the local derivations of $\cO^0_\infty$ and the corresponding Lie
derivatives $d_\la\f=\bL_{d_\la}\f$ of $\cO^*_\infty$.

A derivation $\vt$ (\ref{g3}) is called contact if the Lie
derivative $\bL_\up$ preserves the contact ideal of the DGA
$\cO^*_\infty$, i.e., the Lie derivative $\bL_\up$ of a contact
form is a contact form.

\begin{prop} \label{j11} \mar{j11}
A derivation $\vt$ (\ref{g3}) is contact iff it takes the form
\mar{g4}\beq
\vt=\vt^\la\dr_\la +\vt^i\dr_i +\op\sum_{|\La|>0}
[d_\La(\vt^i-y^i_\m\vt^\m)+y^i_{\m+\La}\vt^\m]. \label{g4}
\eeq
\end{prop}

\begin{proof}
The expression (\ref{g4}) results from a direct computation
similar to that of the first part of B\"acklund's theorem
\cite{ibr}.
\end{proof}

A glance at the expression (\ref{j12}) enables one to regard a
contact derivation (\ref{g4}) as an infinite order jet
prolongation of its restriction
\mar{j15}\beq
\up=\vt^\la\dr_\la +\vt^i\dr_i \label{j15}
\eeq
to the ring $C^\infty(Y)$. Since coefficients $\vt^\la$ and
$\vt^i$ depend on jet coordinates $y^i_\La$, $0<|\La|$, in
general, one calls $\up$ (\ref{j15}) a generalized vector field.
Generalized symmetries of differential equations and Lagrangians
has been intensively studied \cite{and93,ibr,kras,olv}.

Any contact derivation admits the horizontal splitting
\mar{g5}\beq
\vt=\up_H +\up_V=\vt^\la d_\la + [(\vt^i\dr_i + \op\sum_{0<|\La|}
d_\La (\vt^i-y^i_\m\vt^\m)\dr_i^\La] \label{g5}
\eeq
 relative to the canonical connection
$\nabla=dx^\la\ot d_\la$ on the $C^\infty(X)$-ring $\cO^0_\infty$.
One can show \cite{cmp04} that a vertical contact derivation
\be
 \up=\up^i\dr_i +\op\sum_{0<|\La|} d_\La \up^i\dr_i^\La
\ee
obeys the relations
\mar{g6}\beq
\up\rfloor d_H\f=-d_H(\up\rfloor\f), \qquad
\bL_\up(d_H\f)=d_H(\bL_\up\f), \qquad \f\in\cO^*_\infty.
\label{g6}
\eeq
They follow from the equalities
\mar{0480}\ben
&& \up\rfloor \th^i_\La=\up^i_\La, \qquad
d_H(\up^i_\La)=\up^i_{\la+\La}dx^\la, \qquad
d_H\th^i_\la=dx^\la\w\th^i_{\la+\La}, \label{0480}\\
&& d_\la\circ v^i_\La\dr_i^\La= v^i_\La\dr_i^\La \circ d_\la.
\nonumber
\een

\section{Variational bicomplex on fiber bundles}

In order to transform the bicomplex $\cO^{*,*}_\infty$ into the
variational bicomplex, one introduces the $\Bbb R$-module
projector
\mar{r12}\beq
\vr=\op\sum_{k>0} \frac1k\ol\vr\circ h_k\circ h^n, \qquad
\ol\vr(\f)= \op\sum_{|\La|\geq 0} (-1)^{\nm\La}\th^i\w
[d_\La(\dr^\La_i\rfloor\f)], \qquad \f\in \cO^{>0,n}_\infty,
\label{r12}
\eeq
such that $\vr\circ d_H=0$ and the nilpotent variational operator
$\dl=\vr\circ d$ on $\cO^{*,n}_\infty$ which obeys the relation
\beq
\dl\circ\vr-\vr\circ d=0. \label{am13}
\eeq
Let us denote $\bE_k=\vr(\cO^{k,n}_\infty)$. Then the DGA
$\cO^*_\infty$ is split into the variational bicomplex
\mar{7}\beq
\begin{array}{ccccrlcrlccrlccrlcrl}
 & &  &  & & \vdots & & & \vdots  & & &
&\vdots  & & & &
\vdots & &   & \vdots \\
& & & & _{d_V} & \put(0,-7){\vector(0,1){14}} & & _{d_V} &
\put(0,-7){\vector(0,1){14}} & &  & _{d_V} &
\put(0,-7){\vector(0,1){14}} & & &  _{d_V} &
\put(0,-7){\vector(0,1){14}}& & _{-\dl} & \put(0,-7){\vector(0,1){14}} \\
 &  & 0 & \to & &\cO^{1,0}_\infty &\ar^{d_H} & &
\cO^{1,1}_\infty & \ar^{d_H} &\cdots  & & \cO^{1,m}_\infty
&\ar^{d_H} &\cdots & &
\cO^{1,n}_\infty &\ar^\vr &  & \bE_1\to  0\\
& & & & _{d_V} &\put(0,-7){\vector(0,1){14}} & & _{d_V} &
\put(0,-7){\vector(0,1){14}} & & &  _{d_V}
 & \put(0,-7){\vector(0,1){14}} & &  & _{d_V} & \put(0,-7){\vector(0,1){14}}
 & & _{-\dl} & \put(0,-7){\vector(0,1){14}} \\
0 & \to & \Bbb R & \to & & \cO^0_\infty &\ar^{d_H} & &
\cO^{0,1}_\infty & \ar^{d_H} &\cdots  & & \cO^{0,m}_\infty &
\ar^{d_H} & \cdots & &
\cO^{0,n}_\infty & \equiv &  & \cO^{0,n}_\infty \\
& & & & _{\pi^{\infty*}}& \put(0,-7){\vector(0,1){14}} & &
_{\pi^{\infty*}} & \put(0,-7){\vector(0,1){14}} & & &
_{\pi^{\infty*}}
 & \put(0,-7){\vector(0,1){14}} & &  & _{\pi^{\infty*}} &
\put(0,-7){\vector(0,1){14}} & &  & \\
0 & \to & \Bbb R & \to & & \cO^0X &\ar^d & & \cO^1X & \ar^d
&\cdots  & & \cO^mX & \ar^d & \cdots & &
\cO^n(X) & \ar^d & 0 &  \\
& & & & &\put(0,-5){\vector(0,1){10}} & & &
\put(0,-5){\vector(0,1){10}} & & &
 & \put(0,-5){\vector(0,1){10}} & & &   &
\put(0,-5){\vector(0,1){10}} & &  & \\
& & & & &0 & &  & 0 & & & & 0 & & & & 0 & &  &
\end{array}
\label{7}
\eeq
Its relevant cohomology has been obtained as follows
\cite{jmp,ijmms}. One starts from the algebraic Poincar\'e lemma
\cite{olv,tul}.

\begin{lem} \label{042} \mar{042} If $Y$ is a contractible
bundle $\Bbb R^{n+p}\to\Bbb R^n$, the variational bicomplex
(\ref{7}) is exact at all terms, except $\Bbb R$.
\end{lem}

\begin{proof}
The homotopy operators for $d_V$, $d_H$, $\dl$ and $\vr$ are given
by the formulas (5.72), (5.109), (5.84) in \cite{olv} and (4.5) in
\cite{tul}, respectively.
\end{proof}

\begin{theo} \label{g90} \mar{g90}
(i) The second row from the bottom and the last column of this
bicomplex make up the variational complex
\mar{b317}\beq
0\to\Bbb R\to \cO^0_\infty \ar^{d_H}\cO^{0,1}_\infty\cdots
\op\longrightarrow^{d_H} \cO^{0,n}_\infty  \op\longrightarrow^\dl
\bE_1 \op\longrightarrow^\dl \bE_2 \ar \cdots\,. \label{b317}
\eeq
Its cohomology is isomorphic to the de Rham cohomology of the
fiber bundle $Y$, namely,
\mar{j20}\beq
H^{k<n}(d_H;\cQ^*_\infty)=H^{k<n}(Y), \qquad H^{k-n}(\dl;
\cQ^*_\infty)=H^{k\geq n}(Y). \label{j20}
\eeq
(ii) The rows of contact forms of the
bicomplex (\ref{7}) are exact sequences.
\end{theo}

\begin{proof}
Let $\gQ^*_\infty$ be the sheaf of germs of differential forms on
$J^\infty Y$. It is split into the variational bicomplex
$\gQ^{*,*}_\infty$.  Let $\cQ^*_\infty$ be the DGA of global
sections of $\gQ^*_\infty$. It is also decomposed into the
variational bicomplex $\cQ^{*,*}_\infty$. Since the paracompact
space $J^\infty Y$ admits a partition of unity by elements of the
ring $\cQ^0_\infty$, the $d_H$- and $\dl$-cohomology of
$\cQ^{*,*}_\infty)$ can be obtained as follows
\cite{ander,jmp,ijmms,tak2}. Let us consider the variational
subcomplex of $\gQ^{*,*}_\infty$ and the subcomplexes of sheaves
of contact forms
\mar{g91,'}\ben
&& 0\to\Bbb R\to \gQ^0_\infty \ar^{d_H}\gQ^{0,1}_\infty\cdots
\op\longrightarrow^{d_H} \gQ^{0,n}_\infty  \op\longrightarrow^\dl
\gE_1 \op\longrightarrow^\dl \gE_2 \longrightarrow \cdots, \qquad
\gE_k=\vr(\gQ^{k,n}_\infty),
\label{g91} \\
&& 0\to \gQ^{k,0}_\infty\ar^{d_H} \gQ^{k,1}_\infty\cdots \ar^{d_H}
\gQ^{k,n}_\infty  \ar^\vr \gE_k\to 0. \label{g91'}
\een
By virtue of Lemma \ref{042}, these complexes are exact at all
terms, except $\Bbb R$. Since $\gQ^{m,k}_\infty$ are sheaves of
$\cQ^0_\infty$-modules, they are fine. The sheaves $\gE_k$ are
also proved to be fine (see Appendix C). Consequently, all
sheaves, except $\Bbb R$, in the complexes (\ref{g91}) --
(\ref{g91'}) are acyclic. Therefore, these complexes are
resolutions of the constant sheaf $\Bbb R$ and the zero sheaf over
$J^\infty Y$, respectively. Let us consider the corresponding
subcomplexes
\mar{g92,'}\ben
&& 0\to\Bbb R\to \cQ^0_\infty \ar^{d_H}\cQ^{0,1}_\infty\cdots
\op\longrightarrow^{d_H} \cQ^{0,n}_\infty \op\longrightarrow^\dl
\G(\gE_1) \op\longrightarrow^\dl \G(\gE_2) \longrightarrow \cdots,
\label{g92} \\
&& 0\to \cQ^{k,0}_\infty\ar^{d_H} \cQ^{k,1}_\infty\cdots \ar^{d_H}
\cQ^{k,n}_\infty  \ar^\vr \G(\gE_k)\to 0 \label{g92'}
\een
of the DGA $\cQ^*_\infty$. In accordance with the abstract de Rham
theorem (see Appendix B), cohomology of the complex (\ref{g92})
equals the cohomology  of $J^\infty Y$ with coefficients in the
constant sheaf $\Bbb R$, while the complex (\ref{g92'}) is exact.
Since $Y$ is a strong deformation retract of $J^\infty Y$,
cohomology of the complex (\ref{g92}) equals the de Rham
cohomology of $Y$ by virtue of the isomorphisms (\ref{j19}). Note
that, in order to prove the exactness of the complex (\ref{g92'}),
the acyclicity of the sheaves $\gE_k$ need not be justified.
Finally, the subalgebra $\cO^*_\infty\subset \cQ^*_\infty$ is
proved to have the same $d_H$- and $\dl$-cohomology as
$\cQ^*_\infty$ \cite{lmp,ijmms} (see Appendix D). Similarly, one
can show that, restricted to $\cO^{k,n}_\infty$, the operator
$\vr$ remains exact.
\end{proof}

Note that the cohomology isomorphism (\ref{j20}) gives something
more. The relation (\ref{am13}) for $\vr$ and the relation
$h_0d=d_Hh_0$ for $h_0$ define  a cochain morphism of the de Rham
complex (\ref{j4}) of the DGA $\cO^*_\infty$ to its variational
complex (\ref{b317}). The corresponding homomorphism of their
cohomology groups is an isomorphism by virtue of Theorem \ref{j4}
and item (i) of Theorem \ref{g90}. Then the splitting of a closed
form $\f\in\cO^*_\infty$ in Corollary \ref{j21} leads to the
following decompositions.

\begin{prop} \label{t41}
Any $d_H$-closed form $\f\in\cO^{0,m}$, $m< n$, is represented by
a sum
\mar{t60}\beq
\f=h_0\si+ d_H \xi, \qquad \xi\in
\cO^{m-1}_\infty, \label{t60}
\eeq
where $\si$ is a closed $m$-form on $Y$. Any $\dl$-closed form
$\f\in\cO^{k,n}$ is split into
\mar{t42a-c}\ben
&& \f=h_0\si +
d_H\xi, \qquad k=0, \qquad \xi\in \cO^{0,n-1}_\infty,
\label{t42a}\\
&& \f=\vr(\si) +\dl(\xi), \qquad k=1, \qquad \xi\in
\cO^{0,n}_\infty,
\label{t42b}\\
&& \f=\vr(\si) +\dl(\xi), \qquad k>1, \qquad \xi\in \bE_{k-1},
\label{t42c}
\een
where $\si$ is a closed $(n+k)$-form on $Y$.
\end{prop}

One can think of the elements
\be
L=\cL\om\in \cO^{0,n}_\infty, \qquad \dl L=\op\sum_{|\La|\geq
0}(-1)^{|\La|}d_\La(\dr^\La_i \cL)\th^i\w\om\in \bE_1, \qquad
 \om=dx^1\w\cdots\w dx^n,
\ee
of the variational complex (\ref{b317}) as being a finite order
Lagrangian and its Euler--Lagrange operator, respectively. Then
the following are corollaries of Theorem \ref{g90}.

\begin{cor} \label{lmp112'}
(i) A finite order Lagrangian $L\in \cO^{0,n}_\infty$ is
variationally trivial, i.e.,  $\dl(L)=0$ iff
\beq
L=h_0\si + d_H \xi, \qquad \xi\in \cO^{0,n-1}_\infty,
\label{tams3}
\eeq
where $\si$ is a closed $n$-form on $Y$. (ii) A finite order
Euler--Lagrange-type operator $\cE\in \bE_1$ satisfies the
Helmholtz condition $\dl(\cE)=0$ iff
\be
\cE=\dl(L) + \vr(\si), \qquad L\in\cO^{0,n}_\infty,
\ee
where $\si$ is a closed $(n+1)$-form on $Y$.
\end{cor}

\begin{cor} \label{g93} \mar{g93}
The exactness of the row of one-contact forms of the variational
bicomplex (\ref{7}) at the term $\cO^{1,n}_\infty$ relative to the
projector $\vr$ provides the $\Bbb R$-module decomposition
\be
\cO^{1,n}_\infty=\bE_1\oplus d_H(\cO^{1,n-1}_\infty.
\ee
Given a Lagrangian $L\in \cO^{0,n}_\infty$, we have the
corresponding decomposition
\mar{+421}\beq
dL=\dl L-d_H\Xi. \label{+421}
\eeq
\end{cor}

The form $\Xi$ in the decomposition (\ref{+421}) is not uniquely
defined. It reads
\be
\Xi=\op\sum_{s=0}F^{\la\nu_s\ldots\nu_1}_i
\th^i_{\nu_s\ldots\nu_1}\w\om_\la, \quad F_i^{\nu_k\ldots\nu_1}=
\dr_i^{\nu_k\ldots\nu_1}\cL-d_\la F_i^{\la\nu_k\ldots\nu_1}
+h_i^{\nu_k\ldots\nu_1}, \quad \om_\la=\dr_\la\rfloor\om,
\ee
where local functions $h\in\cO^0_\infty$ obey the relations
$h^\nu_i=0$, $h_i^{(\nu_k\nu_{k-1})\ldots\nu_1}=0$. It follows
that $\Xi_L=\Xi +L$ is a Lepagean equivalent of a finite order
Lagrangian $L$ \cite{got}.

The decomposition (\ref{+421}) leads to the global first
variational formula and the first Noether theorem as follows.

\begin{theo}  \label{g75} \mar{g75}
Given a  Lagrangian $L=\cL\om\in\cO^{0,n}_\infty$, its Lie
derivative $\bL_\up L$ along a contact derivation $\up$ (\ref{g5})
fulfils the first variational formula
\mar{g8}\beq
\bL_\vt L= \up_V\rfloor\dl L +d_H(h_0(\vt\rfloor\Xi_L)) +\cL d_V
(\up_H\rfloor\om), \label{g8}
\eeq
where $\Xi_L$ is a Lepagean equivalent.
\end{theo}

\begin{proof}
The formula (\ref{g8}) comes from the splitting (\ref{+421}) and
the relations (\ref{g6}) as follows:
\be
&& \bL_\vt L=\vt\rfloor dL + d(\vt\rfloor L) =[\up_V\rfloor dL
-d_V\cL\w \up_H\rfloor\om] +[d_H(\up_H\rfloor L) + d_V(\cL
\up_H\rfloor\om)]= \\
&& \quad \up_V\rfloor dL + d_H(\up_H\rfloor L) +\cL d_V
(\up_H\rfloor\om)=   \up_V\rfloor\dl L -\up_V\rfloor d_H\Xi +
d_H(\up_H\rfloor L) +\cL d_V (\up_H\rfloor\om)
=  \\
&& \quad \up_V\rfloor\dl L +d_H(\up_V\rfloor\Xi + \up_H\rfloor L)
+\cL d_V (\up_H\rfloor\om),
\ee
where $\up_V\rfloor\Xi=h_0(\vt\rfloor\Xi)$ since $\Xi$ is a
one-contact form, $\up_H\rfloor L=h_0(\up\rfloor L)$, and
$\Xi_L=\Xi+L$.
\end{proof}

A contact derivation $\vt$ (\ref{g4}) is called a variational
symmetry of a Lagrangian $L$ if the Lie derivative $\bL_\vt
L=d_H\xi$ is $d_H$-exact. A glance at the expression (\ref{g8})
shows that: (i) a contact derivation $\vt$ is a variational
symmetry only if it is projected onto $X$ (i.e., its components
$\vt^\la$ depend only on coordinates on $X$), (ii) $\vt$ is a
variational symmetry iff its vertical part $\up_V$ is well, (iii)
it is a variational symmetry iff the density $\up_V\rfloor \dl L$
is $d_H$-exact.

\begin{theo} \label{j22} \mar{j22} If a contact derivation $\vt$
(\ref{g4}) is a variational symmetry of a Lagrangian $L$, the
first variational formula (\ref{g8}) restricted to Ker$\,\dl L$
leads to the weak conservation law
\be
0\ap d_H(h_0(\vt\rfloor\Xi_L)-\xi).
\ee
\end{theo}

\begin{rem} Let a contact derivation $\vt$ (\ref{g4}) be the jet prolongation
of a vector field $\vt^\la \dr_\la +\vt^i\dr_i$ on $Y$. If $\vt$
is a variational symmetry of a Lagrangian $L$, then it is also a
symmetry of the Euler--Lagrange operator $\dl L$ of $L$, i.e.,
$\bL_\vt\dl L=0$ by virtue of the equality $\bL_\vt\dl
L=\dl(\bL_\vt L)$. However, this equality fails to be true in the
case of generalized symmetries \cite{olv}.
\end{rem}

\section{Polynomial variational bicomplex}

Let $Y\to X$ be an affine bundle. Since $X$ is a strong
deformation retract of $Y$, the de Rham cohomology of $Y$ and,
consequently, $J^\infty Y$ equals that of $X$. An immediate
consequence of this fact is the following cohomology isomorphisms
\be
H^{<n}(d_H;\cO^*_\infty)= H^{<n}(X), \qquad
H^0(\dl;\cO^*_\infty)=H^n(X), \qquad H^k(\dl;\cO^*_\infty)=0.
\ee
It follows that every $d_H$-closed form $\f\in \cO^{0,m<n}_\infty$
is represented by the sum
\mar{t2}\beq
\f=\si + d_H\xi, \qquad \xi\in \cO^{0,m-1}_\infty, \label{t2}
\eeq
where $\si$ is a closed form on $X$. Similarly, any variationally
trivial Lagrangian takes the form
\be
L=\si + d_H\xi, \qquad \xi\in \cO^{0,n-1}_\infty,
\ee
where $\si$ is a closed $n$-form on $X$.

Let us restrict our consideration to the short variational complex
\mar{j25}\beq
0\to\Bbb R\to \cO^0_\infty \ar^{d_H}\cO^{0,1}_\infty\cdots
\op\longrightarrow^{d_H} \cO^{0,n}_\infty  \op\longrightarrow^\dl
\bE_1  \label{j25}
\eeq
and the similar complex of sheaves
\mar{j26}\beq
0\to\Bbb R\to \gQ^0_\infty \ar^{d_H}\gQ^{0,1}_\infty\cdots
\op\longrightarrow^{d_H} \gQ^{0,n}_\infty  \op\longrightarrow^\dl
\gE_1. \label{j26}
\eeq
In the case of an affine bundle $Y\to X$, we can lower this
complex  onto the base $X$ as follows.

Let us consider the open surjection $\pi^\infty:J^\infty Y\to X$
and the direct image $\gX^*_\infty=\pi^\infty_*\gQ^*_\infty$ on
$X$ of the sheaf $\gQ^*_\infty$. Its stalk over a point $x\in X$
consists of the equivalence classes of sections of the sheaf
$\gQ^*_\infty$ which coincide on the inverse images
$(\pi^\infty)^{-1}(U_x)$ of neighbourhoods $U_x$ of $x$. Since
$\pi^\infty_*\Bbb R=\Bbb R$, we have the following complex of
sheaves on $X$:
\mar{t71}\beq
0\to\Bbb R\to \gX^0_\infty
\ar^{d_H}\gX^{0,1}_\infty\ar^{d_H}\cdots \op\longrightarrow^{d_H}
\gX^{0,n}_\infty \op\longrightarrow^\dl \pi^\infty_*\gE_1.
\label{t71}
\eeq
Every point $x\in X$ has a base of open contractible
neighbourhoods $\{U_x\}$ such that the sheaves $\gQ^{0,*}_\infty$
of $\cQ^*_\infty$-modules are acyclic on the inverse images
$(\pi^\infty)^{-1}(U_x)$ of these neighbourhoods. Then, in
accordance with the Leray theorem \cite{god}, cohomology of
$J^\infty Y$ with coefficients in the sheaves $\gQ^{0,*}_\infty$
are isomorphic to that of $X$ with coefficients in their direct
images $\gX^{0,*}_\infty$, i.e., the sheaves $\gX^{0,*}_\infty$ on
$X$ are acyclic. Furthermore, Lemma \ref{042} also shows that the
complexes of sections of sheaves $\gQ^{0,*}_\infty$ over
$(\pi^\infty_0)^{-1}(U_x)$ are exact. It follows that the complex
(\ref{t71}) on $X$ is exact at all terms, except $\Bbb R$, and it
is a resolution of the constant sheaf $\Bbb R$ on $X$. Due to the
$\Bbb R$-algebra isomorphism $\cQ^*_\infty=\G(\gX^*_\infty)$, one
can think of the short variational subcomplex of the complex
(\ref{g91}) as being the complex of the structure algebras of the
sheaves in the complex (\ref{t71}) on $X$.

Given the sheaf $\gX^*_\infty$ on $X$, let us consider its
subsheaf $\gP^*_\infty$ of germs of exterior forms which are
polynomials in the fiber coordinates $y^i_\La$, $|\La|\geq 0$, of
the topological fiber bundle $J^\infty Y\to X$. This property is
coordinate-independent due to the transition functions (\ref{j3}).
The sheaf $\gP^*_\infty$ is a sheaf of $C^\infty(X)$-modules. The
DGA $P^*_\infty$ of its global sections is a
$C^\infty(X)$-subalgebra of $\cQ^*_\infty$. We have the subcomplex
\mar{t44}\beq
0\to \Bbb R \ar\gP^0_\infty\ar^{d_H}
\gP^{0,1}_\infty\ar^{d_H}\cdots \ar^{d_H}\gP^{0,n}_\infty\ar^\dl
\pi^\infty_*\gE_1 \label{t44}
\eeq
of the complex (\ref{t71}) on $X$. As a particular variant of the
algebraic Poincar\'e lemma, the exactness of the complex
(\ref{t44}) at all terms, except $\Bbb R$, follows from the form
of the homotopy operator for $d_H$ or can be proved in a
straightforward way \cite{barn}. Since the sheaves
$\gP^{0,*}_\infty$ of $C^\infty(X)$-modules on $X$ are acyclic,
the complex (\ref{t44}) is a resolution of the constant sheaf
$\Bbb R$ on $X$. Hence, cohomology of the complex
\mar{t45}\beq
0\to \Bbb R \ar P^0_\infty\ar^{d_H} P^{0,1}_\infty\ar^{d_H}\cdots
\ar^{d_H} P^{0,n}_\infty\ar^\dl\G(\gE_1) \label{t45}
\eeq
of the DGAs $P^{0,<n}_\infty$ equals the de Rham cohomology of
$X$. It follows that every $d_H$-closed polynomial form $\f\in
P^{0,m<n}_\infty$ is decomposed into the sum
\mar{t72}\beq
\f=\si + d_H\xi, \qquad \xi\in \cP^{0,m-1}_\infty, \label{t72}
\eeq
where $\si$ is a closed form on $X$.

Let $\cP^*_\infty$ be $C^\infty(X)$-subalgebra of the polynomial
algebra $P^*_\infty$ which consists of exterior forms  which are
polynomials in the fiber coordinates $y^i_\La$. Obviously,
$\cP^*_\infty$ is a subalgebra of $\cO^*_\infty$.  Finally, one
can show that $\cP^*_\infty$ have the same cohomology as
$P^*_\infty$, i.e., if $\f$ in the decomposition (\ref{t72}) is an
element of $\cP^{0,*}_\infty$ then $\xi$ is so. The proof of this
fact follows the proof in Appendix D, but differential forms on
$X$ (not $J^\infty Y$) are considered.

\section{Differential calculus on graded manifolds}

We restrict our consideration to graded manifolds $(Z,\gA)$ with
structure sheaves $\gA$ of Grassmann algebras of finite rank
\cite{bart,book05}. By a Grassmann algebra over a ring $\cK$ is
meant a $\Bbb Z_2$-graded exterior algebra of some $\cK$-module.
The symbol $[.]$ stands for the Grassmann parity.

Treating Lagrangian systems of odd variables on a smooth manifold,
we are based on the following variant of the Serre--Swan theorem
\cite{jmp05a}.

\begin{theo} \label{v0} \mar{v0}
Let $Z$ be a smooth manifold. A graded commutative
$C^\infty(Z)$-algebra $\cA$ is isomorphic to the algebra of graded
functions on a graded manifold with a body $Z$ iff it is the
exterior algebra of some projective $C^\infty(Z)$-module of finite
rank.
\end{theo}

\begin{proof} The proof follows at once from the Batchelor theorem
\cite{bart} and the classical Serre--Swan theorem generalized to
an arbitrary smooth manifold \cite{book05,ren}. By virtue of the
first one, any graded manifold $(Z,\gA)$ with a body $Z$ is
isomorphic to the one $(Z,\gA_Q)$, modelled over some vector
bundle $Q\to Z$, whose structure sheaf $\gA_Q$ is the sheaf of
germs of sections of the exterior bundle
\mar{g80}\beq
\w Q^*=\Bbb R\op\oplus_Z Q^*\op\oplus_Z\op\w^2
Q^*\op\oplus_Z\cdots, \label{g80}
\eeq
where $Q^*$ is the dual of $Q\to Z$. The structure ring $\cA_Q$ of
graded functions (sections of $\gA_Q$) on a graded manifold
$(Z,\gA_Q)$ consists of sections of the exterior bundle
(\ref{g80}). The classical Serre--Swan theorem states that a
$C^\infty(Z)$-module is isomorphic to the module of sections of a
smooth vector bundle over $Z$ iff it is a projective module of
finite rank.
\end{proof}

Assuming that Batchelor's isomorphism is fixed from the beginning,
we associate to $(Z,\gA_Q)$ the following DBGA $\cS^*[Q;Z]$
\cite{bart,book05}.  Let us consider the sheaf $\gd\gA_Q$ of
graded derivations of $\gA_Q$. One can show that its sections over
an open subset $U\subset Z$ exhaust all graded derivations of the
graded commutative $\Bbb R$-ring $\cA_U$ of graded functions on
$U$ \cite{bart}. Global sections of $\gd\gA_Q$ make up the real
Lie superalgebra $\gd\cA_Q$ of graded derivations of the $\Bbb
R$-ring $\cA_Q$, i.e.,
\be
u(ff')=u(f)f' +(-1)^{[u][f]}fu(f'), \qquad f,f'\in\cA_Q, \qquad
u\in\cA_Q.
\ee
Then one can construct the Chevalley--Eilenberg complex of
$\gd\cA_Q$ with coefficients in $\cA_Q$ \cite{fuks}. Its
subcomplex $\cS^*[Q;Z]$ of $\cA_Q$-linear morphism is the
Grassmann-graded Chevalley--Eilenberg differential calculus
\mar{v1}\beq
0\to \Bbb R\to \cA_Q \ar^d \cS^1[Q;Z]\ar^d\cdots
\cS^k[Q;Z]\ar^d\cdots \label{v1}
\eeq
over a graded commutative $\Bbb R$-ring $\cA_Q$ \cite{book05}. The
Chevalley--Eilenberg coboundary operator $d$ and the graded
exterior product $\w$ make $\cS^*[Q;Z]$ into a DBGA whose elements
obey the relations
\mar{v21}\beq
\f\w\f' =(-1)^{|\f||\f'| +[\f][\f']}\f'\w \f, \qquad d(\f\w\f')=
d\f\w\f' +(-1)^{|\f|}\f\w d\f'. \label{v21}
\eeq
Given the DGA $\cO^*Z$ of exterior forms on $Z$, there are the
canonical monomorphism $\cO^*Z\to \cS^*[Q;Z]$ and the body
epimorphism $\cS^*[Q;Z]\to \cO^*Z$ which are cochain morphisms.

\begin{lem} \label{v62} \mar{v62}
The DBGA $\cS^*[Q;Z]$ is a minimal differential calculus over
$\cA_Q$, i.e., it is generated by elements $df$, $f\in \cA_Q$.
\end{lem}

\begin{proof}
One can show that elements of $\gd\cA_Q$ are represented by
sections of some vector bundle over $Z$, i.e., $\gd\cA_Q$ is a
projective $C^\infty(Z)$- and $\cA_Q$-module of finite rank, and
so is its $\cA_Q$-dual $\cS^1[Q;Z]$ \cite{cmp04,book05}. Hence,
$\gd\cA_Q$ is the $\cA_Q$-dual of $\cS^1[Q;Z]$ and, consequently,
$\cS^1[Q;Z]$ is generated by elements $df$, $f\in \cA_Q$
\cite{book05}.
\end{proof}

This fact is essential for our consideration because of the
following \cite{book05}.

\begin{lem} \label{v30} \mar{v30}
Given a ring $R$, let $\cK$, $\cK'$ be $R$-rings and  $\cA$,
$\cA'$ the Grassmann algebras over  $\cK$ and $\cK'$,
respectively. Then any homomorphism $\rho: \cA\to \cA'$ yields the
homomorphism of the minimal Chevalley--Eilenberg differential
calculus over a $\Bbb Z_2$-graded $R$-ring $\cA$ to that over
$\cA'$ given by the map $da \mapsto d(\rho(a))$, $a\in\cA$. This
map provides a monomorphism if $\rho$ is a monomorphism of
$R$-algebras
\end{lem}

One can think of elements of the DBGA $\cS^*[Q;Z]$ as being
Grassmann-graded or, simply, graded) differential forms on $Z$ as
follows. Given an open subset $U\subset Z$, let $\cA_U$ be the
Grassmann algebra of sections of the sheaf $\gA_Q$ over $U$, and
let $\cS^*[Q;U]$ be the corresponding Chevalley--Eilenberg
differential calculus over $\cA_U$. Given an open set $U'\subset
U$, the restriction morphisms $\cA_U\to\cA_{U'}$ yield the
restriction morphism  of the DBGAs $\cS^*[Q;U]\to \cS^*[Q;U']$.
Thus,  we obtain the presheaf $\{U,\cS^*[Q;U]\}$ of DBGAs on a
manifold $Z$ and the sheaf $\gS^*[Q;Z]$ of DBGAs of germs of this
presheaf. Since $\{U,\cA_U\}$ is the canonical presheaf of the
sheaf $\gA_Q$, the canonical presheaf of $\gS^*[Q;Z]$ is
$\{U,\cS^*[Q;U]\}$. In particular, $\cS^*[Q;Z]$ is the DBGA of
global sections of the sheaf $\gS^*[Q;Z]$, and  there is the
restriction morphism $\cS^*[Q;Z]\to \cS^*[Q;U]$ for any open
$U\subset Z$.

Due to this restriction morphism, elements of the DBGA
$\cS^*[Q;Z]$ can be written in the following local form. Given
bundle coordinates $(z^A,q^a)$ on $Q$ and the corresponding fiber
basis $\{c^a\}$ for $Q^*\to X$, the tuple $(z^A, c^a)$ is called a
local basis for the graded manifold $(Z,\gA_Q)$ \cite{bart}. With
respect to this basis, graded functions read
\mar{v23}\beq
f=\op\sum_{k=0} \frac1{k!}f_{a_1\ldots a_k}c^{a_1}\cdots c^{a_k},
\label{v23}
\eeq
where $f_{a_1\cdots a_k}$ are smooth real functions on $Z$, and we
omit the symbol of the exterior product of elements $c^a$. Due to
the canonical splitting $VQ= Q\times Q$, the fiber basis
$\{\dr_a\}$ for vertical tangent bundle $VQ\to Q$ of $Q\to Z$ is
the dual of $\{c^a\}$. Then graded derivations take the local form
$u= u^A\dr_A + u^a\dr_a$, where $u^A, u^a$ are local graded
functions. They act on graded functions (\ref{v23}) by the rule
\mar{cmp50'}\beq
u(f_{a\ldots b}c^a\cdots c^b)=u^A\dr_A(f_{a\ldots b})c^a\cdots c^b
+u^d f_{a\ldots b}\dr_d\rfloor (c^a\cdots c^b). \label{cmp50'}
\eeq
Relative to the dual local bases $\{dz^A\}$ for $T^*Z$ and
$\{dc^b\}$ for $Q^*$, graded one-forms read $\f=\f_A dz^A +
\f_adc^a$. The duality morphism is given by the interior product
\be
u\rfloor \f=u^A\f_A + (-1)^{[\f_a]}u^a\f_a, \qquad u\in \gd\cA_Q,
\qquad \f\in \cS^1[Q;Z].
\ee
The Chevalley--Eilenberg coboundary operator $d$, called the
graded exterior differential, reads
\be
d\f=dz^A\w \dr_A\f + dc^a\w \dr_a\f,
\ee
where the derivations $\dr_A$ and $\dr_a$ act on coefficients of
graded differential forms by the formula (\ref{cmp50'}), and they
are graded commutative with the graded differential forms $dz^A$
and $dc^a$.

Since $\cS^*[Q;Z]$ is a DBGA  of graded differential forms on $Z$,
one can obtain its de Rham cohomology by means of the abstract de
Rham theorem as follows.

\begin{theo} \label{j30} \mar {j30}
The cohomology of the de Rham complex (\ref{v1}) of the DBGA
$\cS^*[Q;Z]$ equals the de Rham cohomology of the body $Z$.
\end{theo}

\begin{proof} We have the complex
\mar{1033}\beq
0\to\Bbb R\ar \gS^0[Q;Z] \ar^d \gS^1[Q;Z]\ar^d\cdots
\gS^k[Q;Z]\ar^d\cdots. \label{1033}
\eeq
of sheafs of germs of graded differential forms on $Z$.  Its
members $\gS^k[Q;Z]$ are sheaves of $C^\infty(Z)$-modules on $Z$
and, consequently, are fine and acyclic. Furthermore, the
Poincar\'e lemma for graded differential forms holds \cite{bart}.
It follows that the complex (\ref{1033}) is a fine resolution of
the constant sheaf $\Bbb R$ on the manifold $Z$.  Then, by virtue
of Theorem \ref{+132}, there is an isomorphism
\mar{+136}\beq
H^*(\cS^*[Q;Z])=H^*(Z;\Bbb R)=H^*_{DR}(Z) \label{+136}
\eeq
of the cohomology of the complex (\ref{v1}) to the de Rham
cohomology of $Z$. Moreover, the cohomology isomorphism
(\ref{+136}) accompanies the cochain monomorphism of the de Rham
complex of $\cO^*Z$ to the complex (\ref{v1}). Hence, any closed
graded differential form is split into a sum $\f=\si +d\xi$ of a
closed exterior form $\si$ on $Z$ and an exact graded differential
form.
\end{proof}

\section{Graded infinite order jet manifold}

As was mentioned above, we consider graded manifolds of jets of
smooth fiber bundles, but not jets of fibered graded manifolds. To
motivate this construction, let us return to the case of even
variables when $Y\to X$ is a vector bundle. The jet bundles
$J^kY\to X$ are also vector bundles. Let $\cP^*_\infty \subset
\cO^*_\infty$ be a subalgebra of exterior forms on these bundles
whose coefficients are polynomial in fiber coordinates. In
particular, $\cP^0_\infty$ is the ring of polynomials of these
coordinates with coefficients in the ring $C^\infty(X)$. One can
associate to such a polynomial of degree $m$ a section of the
symmetric product $\op\vee^m (J^kY)^*$ of the dual to some jet
bundle $J^kY\to X$, and {\it vice versa}. Moreover, any element of
$\cP^*_\infty$ is an element of the Chevalley--Eilenberg
differential calculus over $\cP^0_\infty$. Following this example,
let $F\to X$ be a vector bundle, and let us consider graded
manifolds $(X,\cA_{J^rF})$ modelled over the vector bundles
$J^rF\to X$. There is the direct system of the corresponding DBGAs
\be
 \cS^*[F;X]\ar
\cS^*[J^1F;X]\ar\cdots \cS^*[J^rF;X]\ar\cdots,
\ee
whose direct limit $\cS^*_\infty[F;X]$ is the Grassmann-graded
counterpart of an even DGA $\cP^*_\infty$.

In a general setting, let us consider a composite bundle $F\to
Y\to X$ where $F\to Y$ is a vector bundle provided with bundle
coordinates $(x^\la, y^i, q^a)$. Jet manifolds $J^rF$ of $F\to X$
are vector bundles $J^rF\to J^rY$ coordinated by $(x^\la, y^i_\La,
q^a_\La)$, $0\leq |\La|\leq r$. Let $(J^rY,\gA_r)$ be a graded
manifold modelled over this vector bundle. Its local basis is
$(x^\la, y^i_\La, c^a_\La)$, $0\leq|\La|\leq r$. Let
$\cS^*_r[F;Y]$ be the DBGA of graded differential forms on the
graded manifold $(J^rY,\gA_r)$.

There is an epimorphism of graded manifolds $(J^{r+1}Y,\gA_{r+1})
\to (J^rY,\gA_r)$, seen as local-ringed spaces. It consists of the
surjection $\pi^{r+1}_r$ and the sheaf monomorphism
$\pi_r^{r+1*}\gA_r\to \gA_{r+1}$, where $\pi_r^{r+1*}\gA_r$ is the
pull-back onto $J^{r+1}Y$ of the topological fiber bundle
$\gA_r\to J^rY$. This sheaf monomorphism induces the monomorphism
of the canonical presheaves $\ol \gA_r\to \ol \gA_{r+1}$, which
associates to each open subset $U\subset J^{r+1}Y$ the ring of
sections of $\gA_r$ over $\pi^{r+1}_r(U)$. Accordingly, there is
the monomorphism of graded commutative rings $\cA_r \to
\cA_{r+1}$. By virtue of Lemmas \ref{v62} and \ref{v30}, this
monomorphism yields the monomorphism of DBGAs
\mar{v4}\beq
\cS^*_r[F;Y]\to \cS^*_{r+1}[F;Y]. \label{v4}
\eeq
As a consequence, we have the direct system (\ref{j2}) of DBGAs.
Its direct limit $\cS^*_\infty[F;Y]$  is a DBGA of all graded
differential forms $\f\in \cS^*[F_r;J^rY]$ on graded manifolds
$(J^rY,\gA_r)$ modulo monomorphisms (\ref{v4}). Its elements obey
the relations (\ref{v21}).

The monomorphisms $\cO^*_r\to \cS^*_r[F;Y]$ provide a monomorphism
of the direct system (\ref{5.7}) to the direct system (\ref{j2})
and, consequently, the monomorphism
\mar{v7}\beq
\cO^*_\infty Y\to \cS^*_\infty[F;Y] \label{v7}
\eeq
of their direct limits. In particular, $\cS^*_\infty[F;Y]$ is an
$\cO^0_\infty Y$-algebra. Accordingly, the body epimorphisms
$\cS^*_r[F;Y]\to \cO^*_r$ yield the epimorphism of
$\cO^0_\infty$-algebras
\mar{v7'}\beq
\cS^*_\infty[F;Y]\to \cO^*_\infty.  \label{v7'}
\eeq
The morphisms (\ref{v7}0 and (\ref{v7'}) are cochain morphisms
between the de Rham complex (\ref{5.13}) of the DGA $\cO^*_\infty$
and the de Rham complex
\mar{g110}\beq
0\to\Bbb R\ar \cS^0_\infty[F;Y]\ar^d \cS^1_\infty[F;Y]\cdots
\ar^d\cS^k_\infty[F;Y] \ar\cdots \label{g110}
\eeq
of the DBGA $\cS^0_\infty[F;Y]$. Moreover, the corresponding
homomorphisms of cohomology groups of these complexes are
isomorphisms as follows.

\begin{theo} \label{v9} \mar{v9} There is an isomorphism
\mar{v10}\beq
H^*(\cS^*_\infty[F;Y])= H^*(Y) \label{v10}
\eeq
of cohomology $H^*(\cS^*_\infty[F;Y])$ of the de Rham complex
(\ref{g110}) to the de Rham cohomology  $H^*_{DR}(Y)$ of $Y$.
\end{theo}

\begin{proof}
The complex (\ref{g110}) is the direct limit of the de Rham
complexes of the DBGAs $\cS^*_r[F;Y]$. Therefore, the direct limit
of cohomology groups of these complexes is the cohomology of the
de Rham complex (\ref{g110}). By virtue of Theorem \ref{j30},
cohomology of the de Rham complex of $\cS^*_r[F;Y]$ for any $r$
equals the de Rham cohomology of $J^rY$ and, consequently, that of
$Y$, which is the strong deformation retract of any $J^rY$. Hence,
the isomorphism (\ref{v10}) holds.
\end{proof}

It follows that any closed graded differential form $\f\in
\cS^*_\infty[F;Y]$ is split into the sum $\f=d\si +d\xi$ of a
closed exterior form $\si$ on $Y$ and an exact graded differential
form.

One can think of  elements of $\cS^*_\infty[F;Y]$ as being graded
differential forms on the infinite order jet manifold $J^\infty
Y$. Indeed, let $\gS^*_r[F;Y]$ be the sheaf of DBGAs on $J^rY$ and
$\ol\gS^*_r[F;Y]$ its canonical presheaf. Then the above mentioned
presheaf monomorphisms $\ol \gA_r\to \ol \gA_{r+1}$, yield the
direct system of presheaves
\mar{v15}\beq
\ol\gS^*[F;Y]\ar \ol\gS^*_1[F;Y] \ar\cdots \ol\gS^*_r[F;Y]
\ar\cdots, \label{v15}
\eeq
whose direct limit $\ol\gS_\infty^*[F;Y]$ is a presheaf of DBGAs
on the infinite order jet manifold $J^\infty Y$. Let
$\gQ^*_\infty[F;Y]$ be the sheaf of DBGAs of germs of the presheaf
$\ol\gS_\infty^*[F;Y]$.  One can think of the pair $(J^\infty Y,
\gQ^0_\infty[F;Y])$ as being a graded manifold, whose body is the
infinite order jet manifold $J^\infty Y$ and the structure sheaf
$\gQ^0_\infty[F;Y]$ is the sheaf of germs of graded functions on
graded manifolds $(J^rY,\gA_r)$. We agree to call it the graded
infinite order jet manifold. The structure module
$\cQ^*_\infty[F;Y]$ of sections of $\gQ^*_\infty[F;Y]$ is a DBGA
such that, given an element $\f\in \cQ^*_\infty[F;Y]$ and a point
$z\in J^\infty Y$, there exist an open neighbourhood $U$ of $z$
and a graded exterior form $\f^{(k)}$ on some finite order jet
manifold $J^kY$ so that $\f|_U= \pi^{\infty*}_k\f^{(k)}|_U$. In
particular, there is the monomorphism $\cS^*_\infty[F;Y]
\to\cQ^*_\infty[F;Y]$.

Due to this monomorphism, one can restrict $\cS^*_\infty[F;Y]$ to
the coordinate chart (\ref{j3}) and say that $\cS^*_\infty[F;Y]$
as an $\cO^0_\infty Y$-algebra is locally generated by  the
elements
\be
(1, c^a_\La,
dx^\la,\th^a_\La=dc^a_\La-c^a_{\la+\La}dx^\la,\th^i_\La=
dy^i_\La-y^i_{\la+\La}dx^\la), \qquad 0\leq |\La|,
\ee
where $c^a_\La$, $\th^a_\La$ are odd and $dx^\la$, $\th^i_\La$ are
even. We agree to call $(y^i,c^a)$ the local basis for
$\cS^*_\infty[F;Y]$. Let the collective symbol $s^A$ stand for its
elements. Accordingly, the notation $s^A_\La$ and
$\th^A_\La=ds^A_\La- s^A_{\la+\La}dx^\la$ is introduced. For the
sake of simplicity, we further denote $[A]=[s^A]$.

Similarly to $\cO^*_\infty$, the DBGA $\cS^*_\infty[F;Y]$ is
decomposed into $\cS^0_\infty[F;Y]$-modules
$\cS^{k,r}_\infty[F;Y]$ of $k$-contact and $r$-horizontal graded
forms. Accordingly, the graded exterior differential $d$ on
$\cS^*_\infty[F;Y]$ falls into the sum $d=d_H+d_V$ of the total
and vertical differentials, where
\be
d_H(\f)=dx^\la\w d_\la(\f), \qquad d_\la = \dr_\la +
\op\sum_{0\leq|\La|} s^A_{\la+\La}\dr_A^\La.
\ee

Let  $\gd \cS^0_\infty[F;Y]$ be a $\cS^0_\infty[F;Y]$-module of
graded derivation of the $\Bbb R$-ring $\cS^0_\infty[F;Y]$. It is
a real Lie superalgebra. Similarly to Proposition \ref{g62}, one
can show that the DBGA $\cS^*_\infty[F;Y]$ is minimal differential
calculus over the graded commutative $\Bbb R$-ring
$\cS^0_\infty[F;Y]$. The interior product $\vt\rfloor\f$ and the
Lie derivative $\bL_\vt\f$, $\f\in\cS^*_\infty[F;Y]$, $\vt\in \gd
\cS^0_\infty[F;Y]$, obey the relations
\be
&& \vt\rfloor(\f\w\si)=(\vt\rfloor \f)\w\si
+(-1)^{|\f|+[\f][\vt]}\f\w(\vt\rfloor\si), \qquad \f,\si\in
\cS^*_\infty[F;Y] \\
&& \bL_\vt\f=\vt\rfloor d\f+ d(\vt\rfloor\f), \qquad
\bL_\vt(\f\w\si)=\bL_\vt(\f)\w\si
+(-1)^{[\vt][\f]}\f\w\bL_\vt(\si).
\ee
A graded derivation $\vt\in \gd \cS^0_\infty[F;Y]$ is called
contact if the Lie derivative $\bL_\vt$ preserves the ideal of
contact graded forms of the DBGA $\cS^*_\infty[F;Y]$. With respect
to the local basis $(x^\la,s^A_\La, dx^\la,\th^A_\La)$ for the
DBGA $\cS^*_\infty[F;Y]$, any contact graded derivation takes the
form
\mar{g105}\beq
\vt=\up_H+\up_V=\vt^\la d_\la + [\vt^A\dr_A +\op\sum_{|\La|>0}
d_\La(\vt^A-s^A_\m\vt^\m)\dr_A^\La], \label{g105}
\eeq
where $\up_H$ and $\up_V$ denotes its horizontal and vertical
parts. Furthermore, one can justify that any vertical contact
graded derivation
\mar{j40}\beq
\vt= \vt^A\dr_A +\op\sum_{|\La|>0} d_\La\vt^A\dr_A^\La \label{j40}
\eeq
satisfies the relations
\mar{g233}\beq
 \vt\rfloor d_H\f=-d_H(\vt\rfloor\f), \qquad
 \bL_\vt(d_H\f)=d_H(\bL_\vt\f), \qquad \f\in\cS^*_\infty[F;Y].
\label{g233}
\eeq

\section{Grassmann-graded variational bicomplex}

Similarly to the DGA $\cO^*_\infty$, the DBGA $\cS^*_\infty[F;Y]$
is provided with the graded projection endomorphism
\be
\vr=\op\sum_{k>0} \frac1k\ol\vr\circ h_k\circ h^n, \qquad
\ol\vr(\f)= \op\sum_{0\leq|\La|} (-1)^{\nm\La}\th^A\w
[d_\La(\dr^\La_A\rfloor\f)], \qquad \f\in \cS^{>0,n}_\infty[F;Y],
\ee
such that $\vr\circ d_H=0$ and the nilpotent graded variational
operator $\dl=\vr\circ d$. With these operators the bicomplex BGDA
$\cS^{*,}_\infty[F;Y]$ is completed to the Grassmann-graded
variational bicomplex. We restrict our consideration to its short
variational subcomplex
\mar{g111}\beq
0\ar \Bbb R\ar \cS^0_\infty[F;Y]\ar^{d_H}\cS^{0,1}_\infty[F;Y]
\cdots \ar^{d_H} \cS^{0,n}_\infty[F;Y]\ar^\dl \bE_1, \quad
\bE_1=\vr(\cS^{1,n}_\infty[F;Y]), \label{g111}
\eeq
and its subcomplex of one-contact graded forms
\mar{g112}\beq
 0\to \cS^{1,0}_\infty[F;Y]\ar^{d_H} \cS^{1,1}_\infty[F;Y]\cdots
\ar^{d_H}\cS^{1,n}_\infty[F;Y]\ar^\vr \bE_1\to 0. \label{g112}
\eeq
One can think of its even elements
\mar{0709}\beq
L=\cL\om\in \cS^{0,n}_\infty[F;Y], \qquad \dl L= \th^A\w
\cE_A\om=\op\sum_{0\leq|\La|}
 (-1)^{|\La|}\th^A\w d_\La (\dr^\La_A L)\om\in \bE_1 \label{0709'}
\eeq
as being a Grassmann-graded Lagrangian and its Euler--Lagrange
operator, respectively.

\begin{theo} \label{v11} \mar{v11}
Cohomology of the complex (\ref{g111}) equals the de Rham
cohomology $H^*_{DR}(Y)$ of $Y$. The complex (\ref{g112}) is
exact.
\end{theo}

The proof of Theorem \ref{v11} follows the scheme of the proof of
Theorem \ref{g90}. It falls into three steps.

(i) We start with showing that the complexes (\ref{g111}) --
(\ref{g112}) are locally exact.

\begin{lem} \label{0465} \mar{0465}
If $Y=\Bbb R^{n+k}\to \Bbb R^n$, the complex (\ref{g111}) at all
terms, except $\Bbb R$, is exact.
\end{lem}

\begin{proof}
Referring to  \cite{barn,drag} for the proof, we summarize a few
formulas. Any horizontal graded form $\f\in \cS^{0,*}_\infty$
admits the decomposition
\mar{0471}\beq
\f=\f_0 + \wt\f, \qquad \wt\f=
\op\int^1_0\frac{d\la}{\la}\op\sum_{0\leq|\La|}s^A_\La
\dr^\La_A\f, \label{0471}
\eeq
where $\f_0$ is an exterior form on $\Bbb R^{n+k}$. Let $\f\in
\cS^{0,m<n}_\infty$ be $d_H$-closed.  Then its component $\f_0$
(\ref{0471}) is an exact exterior form on $\Bbb R^{n+k}$  and
$\wt\f=d_H\xi$, where $\xi$ is given by the following expressions.
Let us introduce the operator
\mar{0470}\beq
D^{+\nu}\wt\f=\op\int^1_0\frac{d\la}{\la}\sum_{0\leq k}
k\dl^\nu_{(\m_1}\dl^{\al_1}_{\m_2}\cdots\dl^{\al_{k-1}}_{\m_k)}
\la s^A_{(\al_1\ldots\al_{k-1})}
\dr_A^{\m_1\ldots\m_k}\wt\f(x^\m,\la s^A_\La, dx^\m). \label{0470}
\eeq
The relation $[D^{+\nu},d_\m]\wt\f=\dl^\nu_\m\wt\f$ holds, and
leads to the desired expression
\mar{0473}\beq
\xi=\op\sum_{k=0}\frac{(n-m-1)!}{(n-m+k)!}D^{+\nu} P_k
\dr_\nu\rfloor\wt\f, \qquad P_0=1, \quad
 P_k=d_{\nu_1}\cdots d_{\nu_k}D^{+\nu_1}\cdots
D^{+\nu_k}.
 \label{0473}
\eeq
Now let $\f\in \cS^{0,n}_\infty$ be a graded density such that
$\dl\f=0$. Then its component $\f_0$ (\ref{0471}) is an exact
$n$-form on $\Bbb R^{n+k}$ and $\wt\f=d_H\xi$, where $\xi$ is
given by the expression
\mar{0474}\beq
\xi=\op\sum_{|\La|\geq 0}\op\sum_{\Si+\Xi=\La}(-1)^{|\Si|}s^A_\Xi
d_\Si\dr^{\m+\La}_A\wt\f\om_\m. \label{0474}
\eeq
Since elements of $\cS^*_\infty$ are polynomials in $s^A_\La$, the
sum in the expression (\ref{0473}) is finite. However, the
expression (\ref{0473}) contains a $d_H$-exact summand which
prevents its extension to $\cO^*_\infty$. In this respect, we also
quote the homotopy operator (5.107) in \cite{olv} which leads to
the expression
\mar{0477}\ben
&&\xi=\op\int_0^1 I(\f)(x^\m,\la s^A_\La,
dx^\m)\frac{d\la}{\la}, \label{0477}\\
&& I(\f)=\op\sum_{0\leq|\La|}\op\sum_\m
\frac{\La_\m+1}{n-m+|\La|+1} d_\La[ \op\sum_{0\leq|\Xi|} (-1)^\Xi
\frac{(\m+\La+\Xi)!}{(\m+\La)!\Xi!}s^A d_\Xi
\dr_A^{\m+\La+\Xi}(\dr_\m\rfloor \f)], \nonumber
\een
where $\La!=\La_{\m_1}!\cdots \La_{\m_n}!$ and $\La_\m$ denotes
the number of occurrences of the index $\m$ in $\La$ \cite{olv}.
The graded forms (\ref{0474}) and (\ref{0477}) differ in a
$d_H$-exact graded form.
\end{proof}

\begin{lem} \label{g220} \mar{g220}
If $Y=\Bbb R^{n+k}\to \Bbb R^n$, the complex (\ref{g112}) is
exact.
\end{lem}

\begin{proof}
The fact that a $d_H$-closed graded $(1,m)$-form $\f\in
\cS^{1,m<n}_\infty$ is $d_H$-exact is derived from Lemma
\ref{0465} as follows. We write
\mar{0445}\beq
\f=\sum\f_A^\La\w \th^A_\La, \label{0445}
\eeq
where $\f_A^\La\in \cS^{0,m}_\infty$ are horizontal graded
$m$-forms. Let us introduce additional variables $\ol s^A_\La$ of
the same Grassmann parity as $s^A_\La$. Then one can associate to
each graded $(1,m)$-form $\f$ (\ref{0445}) a unique horizontal
graded $m$-form
\mar{0446}\beq
\ol\f=\sum\f_A^\La\ol s^A_\La, \label{0446}
\eeq
whose coefficients are linear in the variables $\ol s^A_\La$, and
{\it vice versa}.  Let us consider the modified total differential
\be
\ol d_H=d_H + dx^\la\w \op\sum_{0<|\La|}\ol
s^A_{\la+\La}\ol\dr_A^\La,
\ee
acting on graded forms (\ref{0446}), where $\ol\dr^\La_A$ is the
dual of $d\ol s^A_\La$. Comparing the equality $\ol d_H\ol
s^A_\La=dx^\la s^A_{\la+\La}$ and the last equality (\ref{0480}),
one can  easily justify that $\ol{d_H\f}=\ol d_H\ol\f$. Let a
graded $(1,m)$-form $\f$ (\ref{0445}) be $d_H$-closed. Then the
associated horizontal graded $m$-form $\ol \f$ (\ref{0446}) is
$\ol d_H$-closed and, by virtue of Lemma \ref{0465}, it is $\ol
d_H$-exact, i.e., $\ol \f= \ol d_H \ol\xi$, where $\ol\xi$ is a
horizontal graded $(m-1)$-form given by the expression
(\ref{0473}) depending  on additional variables $\ol s^A_\La$. A
glance at this expression shows that, since $\ol\f$ is linear in
the variables $\ol s^A_\La$, so is $\ol\xi=\sum\xi_A^\La\ol
s^A_\La$. It follows that $\f=d_H\xi$ where $\xi=\sum\xi_A^\La\w
\th^A_\La$. It remains to prove the exactness of the complex
(\ref{g112}) at the last term $\bE_1$. If
\be
\vr(\si)=\op\sum_{0\leq|\La|}(-1)^{|\La|}\th^A\w
[d_\La(\dr_A^\La\rfloor\si)]= \op\sum_{0\leq|\La|}
(-1)^{|\La|}\th^A\w [d_\La\si_A^\La]\om=0, \qquad  \si\in
\cS^{1,n}_\infty,
\ee
a direct computation gives
\mar{0449}\beq
\si=d_H\xi,\qquad  \xi=-\op\sum_{0\leq|\La|}\op\sum_{\Si+\Xi=\La}
(-1)^{|\Si|}\th^A_{\Xi}\w d_\Si\si^{\m+\La}_A \om_\m. \label{0449}
\eeq
\end{proof}

\begin{rem}
The proof of Lemma \ref{g220} fails to be extended to complexes of
higher contact forms because the products $\th^A_\La\w\th^B_\Si$
and $s^A_\La s^B_\Si$ obey different commutation rules.
\end{rem}

(ii) Let us now prove Theorem \ref{v11} for the DBGA
$\cQ^*_\infty[F;Y]$. Similarly to $\cS^*_\infty[F;Y]$, the sheaf
$\gQ^*_\infty[F;Y]$ and the DBGA $\cQ^*_\infty[F;Y]$ are split
into the Grassmann-graded variational bicomplexes. We consider
their subcomplexes
\mar{v35-8}\ben
&& 0\ar \Bbb R\ar \gQ^0_\infty[F;Y]\ar^{d_H}\gQ^{0,1}_\infty[F;Y]
\cdots \ar^{d_H} \gQ^{0,n}_\infty[F;Y]\ar^\dl {\got E}_1,
\label{v35}\\
&& 0\to \gQ^{1,0}_\infty[F;Y]\ar^{d_H} \gQ^{1,1}_\infty[F;Y]\cdots
\ar^{d_H}\gQ^{1,n}_\infty[F;Y]\ar^\vr {\got E}_1\to 0, \label{v36}\\
&& 0\ar \Bbb R\ar \cQ^0_\infty[F;Y]\ar^{d_H}\cQ^{0,1}_\infty[F;Y]
\cdots \ar^{d_H} \cQ^{0,n}_\infty[F;Y]\ar^\dl \G({\got E}_1),
\label{v37} \\
&&  0\to \cQ^{1,0}_\infty[F;Y]\ar^{d_H}
\cQ^{1,1}_\infty[F;Y]\cdots \ar^{d_H}\cQ^{1,n}_\infty[F;Y]\ar^\vr
\G({\got E}_1)\to 0, \label{v38}
\een
where ${\got E}_1 =\vr(\gQ^{1,n}_\infty[F;Y])$. By virtue of
Lemmas \ref{0465} and \ref{g220}, the complexes (\ref{v35}) --
(\ref{v36}) at all terms, except $\Bbb R$, are exact. The terms
$\gQ^{*,*}_\infty[F;Y]$ of the complexes (\ref{v35}) --
(\ref{v36}) are sheaves of $\cQ^0_\infty$-modules. Since $J^\infty
Y$ admits the partition of unity just by elements of
$\cQ^0_\infty$, these sheaves are fine and, consequently, acyclic.
By virtue of the abstract de Rham theorem (see Appendix B),
cohomology of the complex (\ref{v37}) equals the cohomology of
$J^\infty Y$ with coefficients in the constant sheaf $\Bbb R$ and,
consequently, the de Rham cohomology of $Y$ in accordance with
isomorphisms (\ref{j19}). Similarly, the complex (\ref{v38}) is
proved to be exact.

(iii) It remains to prove that cohomology of the complexes
(\ref{g111}) -- (\ref{g112}) equals that of the complexes
(\ref{v37}) -- (\ref{v38}). The proof of this fact
straightforwardly follows the proof of Theorem \ref{g90}, and it
is a slight modification of the proof of \cite{cmp04}, Theorem
4.1, where graded exterior forms on the infinite order jet
manifold $J^\infty Y$ of an affine bundle are treated as those on
$X$.

\begin{prop} \label{cmp26} \mar{cmp26}
Every $d_H$-closed graded form $\f\in\cS^{0,m<n}_\infty[F;Y]$
falls into the sum
\mar{g214}\beq
\f=h_0\si + d_H\xi, \qquad \xi\in \cS^{0,m-1}_\infty[F;Y],
\label{g214}
\eeq
where $\si$ is a closed $m$-form on $Y$. Any $\dl$-closed graded
density (e.g., a variationally trivial Grassmann-graded
Lagrangian) $L\in \cS^{0,n}_\infty[F;Y]$ is the sum
\mar{g215}\beq
L=h_0\si + d_H\xi, \qquad \xi\in \cS^{0,n-1}_\infty[F;Y],
\label{g215}
\eeq
where $\si$ is a closed $n$-form on $Y$. In particular, an odd
$\dl$-closed graded density is always $d_H$-exact.
\end{prop}

\begin{proof}
The complex (\ref{g111}) possesses the same cohomology as the
short variational complex
\mar{b317'}\beq
0\to\Bbb R\to \cO^0_\infty  \ar^{d_H}\cO^{0,1}_\infty \cdots
\op\ar^{d_H} \cO^{0,n}_\infty  \op\ar^\dl \bE_1 \label{b317'}
\eeq
of the DGA  $\cO^*_\infty$. The monomorphism (\ref{v7}) and the
body epimorphism (\ref{v7'}) yield the corresponding cochain
morphisms of the complexes (\ref{g111}) and (\ref{b317'}).
Therefore, cohomology of the complex (\ref{g111}) is the image of
the cohomology of $\cO^*_\infty$.
\end{proof}

The global exactness of the complex (\ref{g112}) at the term
$\cS^{1,n}_\infty[F;Y]$ results in the following \cite{cmp04}.

\begin{prop} \label{g103} \mar{g103}
Given a Grassmann-graded Lagrangian $L=\cL\om$, there is the
decomposition
\mar{g99,'}\ben
&& dL=\dl L - d_H\Xi,
\qquad \Xi\in \cS^{1,n-1}_\infty[F;Y], \label{g99}\\
&& \Xi=\op\sum_{s=0} \th^A_{\nu_s\ldots\nu_1}\w
F^{\la\nu_s\ldots\nu_1}_A\om_\la,\qquad F_A^{\nu_k\ldots\nu_1}=
\dr_A^{\nu_k\ldots\nu_1}\cL-d_\la F_A^{\la\nu_k\ldots\nu_1}
+h_A^{\nu_k\ldots\nu_1},  \label{g99'}
\een
where local graded functions $h$ obey the relations $h^\nu_a=0$,
$h_a^{(\nu_k\nu_{k-1})\ldots\nu_1}=0$.
\end{prop}

Note that, locally, one can always choose $\Xi$ (\ref{g99'}) where
all functions $h$ vanish.

The decomposition (\ref{g99}  leads to the global first
variational formula for Grassmann-graded Lagrangians as follows
\cite{jmp05,cmp04}.

\begin{prop} \label{j44} \mar{j44}
Let $\vt\in\gd \cS^0_\infty[F;Y]$ be a contact graded derivation
(\ref{g105}) of the $\Bbb R$-ring $\cS^0_\infty[F;Y]$. Then the
Lie derivative $\bL_\vt L$ of a Lagrangian $L$ fulfills the first
variational formula
\mar{g107}\beq
\bL_\vt L= \vt_V\rfloor\dl L +d_H(h_0(\vt\rfloor \Xi_L)) + d_V
(\vt_H\rfloor\om)\cL, \label{g107}
\eeq
where $\Xi_L=\Xi+L$ is a Lepagean equivalent of $L$ given by the
coordinate expression (\ref{g99'}).
\end{prop}

\begin{proof}
The proof follows from the splitting (\ref{g99}) similarly to the
proof of Proposition \ref{g75}.
\end{proof}

A contact graded derivation $\vt$ (\ref{g105}) is called a
variational symmetry of a Lagrangian $L$ if the Lie derivative
$\bL_\vt L=d_H\xi$ is $d_H$-exact. A glance at the expression
(\ref{g107}) shows that: (i) a contact graded derivation $\vt$ is
a variational symmetry only if it is projected onto $X$, (ii)
$\vt$ is a variational symmetry iff its vertical part $\up_V$ is
well, (iii) it is a variational symmetry iff the density
$\up_V\rfloor \dl L$ is $d_H$-exact.

\begin{theo} \label{j45} \mar{j45} If a contact graded derivation $\vt$
(\ref{g105}) is a variational symmetry of a Lagrangian $L$, the
first variational formula (\ref{g8}) restricted to Ker$\,\dl L$
leads to the weak conservation law
\be
0\ap d_H(h_0(\vt\rfloor\Xi_L)-\xi).
\ee
\end{theo}

\begin{rem}
If $Y\to X$ is an affine bundle, one can consider the subalgebra
$\cP[F;Y]\subset \cS[F;Y]$ of graded differential forms whose
coefficients are polynomials in fiber coordinates of $Y\to X$ and
their jets. This subalgebra is also split into the
Grassmann-graded variational bicomplex. One can show that, the
cohomology of its short variational subcomplex as like as that of
the complex  (\ref{t45}) equals the de Rham cohomology of $X$.
\end{rem}

\section{Appendixes}

{\it Appendix A.} To show that $Y$ is a strong deformation retract
of $J^\infty Y$, let us construct a homotopy from $J^\infty Y$ to
$Y$ in an explicit form. Let $\g_{(k)}$, $k\leq 1$, be global
sections of the affine jet bundles $J^kY\to J^{k-1}Y$. Then, we
have a global section
\mar{am4}\beq
\g: Y \ni (x^\la,y^i)\to (x^\la,y^i, y^i_\La =\g_{(|\La|)}{}^i_\La
\circ\g_{(|\La|-1)}\circ \cdots \circ \g_{(1)}) \in J^\infty Y.
\label{am4}
\eeq
of the open surjection $\pi^\infty_0: J^\infty Y\to Y$. Let us
consider the map
\mar{am5}\ben
&&[0,1]\times J^\infty Y\ni (t; x^\la, y^i, y^i_\La) \to (x^\la,
y^i,y'^i_\La)\in J^\infty Y, \qquad 0<|\La|, \label{am5}\\
&& y'^i_\La= f_k(t)y^i_\La
+(1-f_k(t))\g_{(k)}{}^i_\La(x^\la,y^i,y^i_\Si), \qquad
|\Si|<k=|\La|, \nonumber
\een
where $f_k(t)$ is a continuous monotone real function on $[0,1]$
such that
\mar{am6}\beq
f_k(t)=\left\{
\begin{array}{ll}
0, & \quad t\leq 1-2^{-k},\\
1, & \quad t\geq 1-2^{-(k+1)}.
\end{array}\right. \label{am6}
\eeq
A glance at the transition functions (\ref{j3}) shows that,
although written in a coordinate form, this map is globally
defined. It is continuous because, given an open subset
$U_k\subset J^kY$, the inverse image of the open set
$(\pi^\infty_k)^{-1}(U_k)\subset J^\infty Y$, is the open subset
 \be
&& (t_k,1]\times (\pi^\infty_k)^{-1}(U_k) \cup (t_{k-1},1]\times
(\pi^\infty_{k-1})^{-1}(\pi^k_{k-1} [U_k\cap
\g_{(k)}(J^{k-1}Y)])\cup\cdots\\
&& \qquad \cup [0,1]\times (\pi^\infty_0)^{-1}(\pi^k_0 [U_k\cap
\g_{(k)}\circ\cdots\circ\g_{(1)}(Y)])
\ee
of $[0,1]\times J^\infty Y$, where $[t_r,1]={\rm supp}\,f_r$.
Then, the map (\ref{am5}) is a desired homotopy from $J^\infty Y$
to $Y$ which is identified with its image under the global section
(\ref{am4}).

\noindent {\it Appendix B.} We quote the following minor
generalization of the abstract de Rham theorem (\cite{hir},
Theorem 2.12.1) \cite{jmp,tak2}. Let
\be
0\to S\ar^h S_0\ar^{h^0} S_1\ar^{h^1}\cdots\ar^{h^{p-1}}
S_p\ar^{h^p} S_{p+1}, \qquad p>1,
\ee
be an exact sequence of sheaves of Abelian groups over a
paracompact topological space $Z$, where the sheaves $S_q$, $0\leq
q<p$, are acyclic, and let
\beq
0\to \G(Z,S)\ar^{h_*} \G(Z,S_0)\ar^{h^0_*}
\G(Z,S_1)\ar^{h^1_*}\cdots\ar^{h^{p-1}_*} \G(Z,S_p)\ar^{h^p_*}
\G(Z,S_{p+1}) \label{+130}
\eeq
be the corresponding cochain complex of sections of these sheaves.

\begin{theo} \label{+132} \mar{+132}
The $q$-cohomology groups of the cochain complex (\ref{+130}) for
$0\leq q\leq p$ are isomorphic to the cohomology groups $H^q(Z,S)$
of $Z$ with coefficients in the sheaf $S$.
\end{theo}

\bigskip

\noindent {\it Appendix C.} The sheaves $\gE_k$ in proof of
Theorem \ref{g90} are fine as follows \cite{jmp}. Though the $\Bbb
R$-modules $\G(\gE_{k>1})$ fail to be $\cQ^0_\infty$-modules
\cite{tul}, one can use the fact that the sheaves $\gE_{k>0}$ are
projections $\vr(\gQ^{k,n}_\infty)$ of sheaves of
$\cQ^0_\infty$-modules. Let $\{U_i\}_{i\in I}$ be a locally finite
open covering  of $J^\infty Y$ and $\{f_i\in\cQ^0_\infty\}$ the
associated partition of unity. For any open subset $U\subset
J^\infty Y$ and any section $\varphi$ of the sheaf
$\gQ^{k,n}_\infty$ over $U$, let us put $h_i(\varphi)=f_i\varphi$.
The endomorphisms $h_i$ of $\gQ^{k,n}_\infty$ yield the $\Bbb
R$-module endomorphisms
\be
\ol h_i=\vr\circ h_i: \gE_k\ar^{\rm in} \gQ^{k,n}_\infty \ar^{h_i}
\gQ^{k,n}_\infty \ar^\vr \gE_k
\ee
of the sheaves $\gE_k$. They possess the properties required for
$\gE_k$ to be a fine sheaf. Indeed, for each $i\in I$, ${\rm
supp}\,f_i\subset U_i$ provides a closed set  such that $\ol h_i$
is zero outside this set, while the sum $\op\sum_{i\in I}\ol h_i$
is the identity morphism.

\bigskip

\noindent {\it Appendix D.} Let the common symbol $D$ stand for
$d_H$ and $\dl$. Bearing in mind decompositions (\ref{t60}) --
(\ref{t42c}), it suffices to show that, if an element $\f\in
\cO^*_\infty$ is $D$-exact in the algebra $\cQ^*_\infty$, then it
is so in the algebra $\cO^*_\infty$. Lemma \ref{042} states that,
if $Y$ is a contractible bundle and a $D$-exact form $\f$ on
$J^\infty Y$ is of finite jet order $[\f]$ (i.e.,
$\f\in\cO^*_\infty$), there exists a differential form $\varphi\in
\cO^*_\infty$ on $J^\infty Y$ such that $\f=D\varphi$. Moreover, a
glance at the homotopy operators for $d_H$ and $\dl$ shows that
the jet order $[\varphi]$ of $\varphi$ is bounded by an integer
$N([\f])$, depending only on the jet order of $\f$. Let us call
this fact the finite exactness of the operator $D$. Given an
arbitrary bundle $Y$, the finite exactness takes place on
$J^\infty Y|_U$ over any domain $U\subset Y$. Let us prove the
following.

(i) Given a family $\{U_\al\}$ of disjoint open subsets of $Y$,
let us suppose that the finite exactness takes place on $J^\infty
Y|_{U_\al}$ over every subset $U_\al$ from this family. Then, it
is true on $J^\infty Y$ over the union $\op\cup_\al U_\al$ of
these subsets.

(ii) Suppose that the finite exactness of the operator $D$ takes
place on $J^\infty Y$ over open subsets $U$, $V$ of $Y$ and their
non-empty overlap $U\cap V$. Then, it is also true on $J^\infty
Y|_{U\cup V}$.

{\it Proof of (i)}. Let $\f\in\cO^*_\infty$ be a $D$-exact form on
$J^\infty Y$. The finite exactness on $(\pi^\infty_0)^{-1}(\cup
U_\al)$ holds since $\f=D\varphi_\al$ on every
$(\pi^\infty_0)^{-1}(U_\al)$ and $[\varphi_\al]< N([\f])$.

{\it Proof of (ii)}. Let $\f=D\varphi\in\cO^*_\infty$ be a
$D$-exact form on $J^\infty Y$. By assumption, it can be brought
into the form $D\varphi_U$ on $(\pi^\infty_0)^{-1}(U)$ and
$D\varphi_V$ on $(\pi^\infty_0)^{-1}(V)$, where $\varphi_U$ and
$\varphi_V$ are differential forms of bounded jet order. Let us
consider their difference $\varphi_U-\varphi_V$ on
$(\pi^\infty_0)^{-1}(U\cap V)$. It is a $D$-exact form of bounded
jet order $[\varphi_U-\varphi_V]< N([\f])$ which, by assumption,
can be written as $\varphi_U-\varphi_V=D\si$ where $\si$ is also
of bounded jet order $[\si]<N(N([\f]))$. Lemma \ref{am20} below
shows that $\si=\si_U +\si_V$ where $\si_U$ and $\si_V$ are
differential forms of bounded jet order on
$(\pi^\infty_0)^{-1}(U)$ and $(\pi^\infty_0)^{-1}(V)$,
respectively. Then, putting
\be
\varphi'|_U=\varphi_U-D\si_U, \qquad \varphi'|_V=\varphi_V+
D\si_V,
\ee
we have the form $\f$, equal to $D\varphi'_U$ on
$(\pi^\infty_0)^{-1}(U)$ and $D\varphi'_V$ on
$(\pi^\infty_0)^{-1}(V)$, respectively. Since the difference
$\varphi'_U -\varphi'_V$ on $(\pi^\infty_0)^{-1}(U\cap V)$
vanishes, we obtain $\f=D\varphi'$ on $(\pi^\infty_0)^{-1}(U\cup
V)$ where
\be
\varphi'\op=^{\rm def}\left\{
\begin{array}{ll}
\varphi'|_U=\varphi'_U, &\\
\varphi'|_V=\varphi'_V &
\end{array}\right.
\ee
is of bounded jet order $[\varphi']<N(N([\f]))$.

To prove the finite exactness of $D$ on $J^\infty Y$, it remains
to choose an appropriate cover of $Y$. A smooth manifold $Y$
admits a countable cover $\{U_\xi\}$ by domains $U_\xi$, $\xi\in
{\bf N}$, and its refinement $\{U_{ij}\}$, where $j\in {\bf N}$
and $i$ runs through a finite set, such that $U_{ij}\cap
U_{ik}=\emptyset$, $j\neq k$ \cite{greub}. Then $Y$ has a finite
cover $\{U_i=\cup_j U_{ij}\}$. Since the finite exactness of the
operator $D$ takes place over any domain $U_\xi$, it also holds
over any member $U_{ij}$ of the refinement $\{U_{ij}\}$ of
$\{U_\xi\}$ and, in accordance with item (i) above, over any
member of the finite cover $\{U_i\}$ of $Y$. Then by virtue of
item (ii) above, the finite exactness of $D$ takes place over $Y$.

\begin{lem} \label{am20}
Let $U$ and $V$ be open subsets of a bundle $Y$ and $\si\in
\gO^*_\infty$ a differential form of bounded jet order on
$(\pi^\infty_0)^{-1}(U\cap V)\subset J^\infty Y$. Then, $\si$ is
split into  a sum $\si_U+ \si_V$ of differential forms $\si_U$ and
$\si_V$ of bounded jet order on $(\pi^\infty_0)^{-1}(U)$ and
$(\pi^\infty_0)^{-1}(V)$, respectively.
\end{lem}

\begin{proof}
By taking a smooth partition of unity on $U\cup V$ subordinate to
the cover $\{U,V\}$ and passing to the function with support in
$V$, one gets a smooth real function $f$ on $U\cup V$ which is 0
on a neighborhood of $U-V$ and 1 on a neighborhood of $V-U$ in
$U\cup V$. Let $(\pi^\infty_0)^*f$ be the pull-back of $f$ onto
$(\pi^\infty_0)^{-1}(U\cup V)$. The differential form
$((\pi^\infty_0)^*f)\si$ is 0 on a neighborhood of
$(\pi^\infty_0)^{-1}(U)$ and, therefore, can be extended by 0 to
$(\pi^\infty_0)^{-1}(U)$. Let us denote it $\si_U$. Accordingly,
the differential form $(1-(\pi^\infty_0)^*f)\si$ has an extension
$\si_V$ by 0 to $(\pi^\infty_0)^{-1}(V)$. Then, $\si=\si_U +\si_V$
is a desired decomposition because $\si_U$ and $\si_V$ are of the
jet order which does not exceed that of $\si$.
\end{proof}

\end{document}